\documentclass[onefignum,onetabnum]{siamart171218}

\usepackage{xcolor}
\usepackage{graphics}
\usepackage{graphicx}
\usepackage{todonotes}
\usepackage{setspace}

\newcommand{\Div}{\mbox{div}}
\newcommand{\Span}{\mbox{span}}

\newcommand{\bv}{{\bf b}}

\newcommand{\dv}{{\bf d}}
\newcommand{\ev}{{\bf e}}
\newcommand{\fv}{{\bf f}}
\newcommand{\gvv}{{\bf g}}
\newcommand{\hv}{{\bf h}}

\newcommand{\pv}{{\bf p}}
\newcommand{\qv}{{\bf q}}
\newcommand{\rv}{{\bf r}}

\newcommand{\uv}{{\bf u}}
\newcommand{\vvv}{{\bf v}}
\newcommand{\wv}{{\bf w}}

\newcommand{\yv}{{\bf y}}
\newcommand{\zv}{{\bf z}}

\newcommand{\Qm}{{\bf Q}}

\newcommand{\Nm}{{\bf N}}
\newcommand{\Mm}{{\bf M}}
\newcommand{\Am}{{\bf A}}
\newcommand{\Bm}{{\bf B}}

\newcommand{\Id}{{\bf I}}
\newcommand{\Vm}{{\bf V}}

\newcommand{\Wm}{{\bf W}}
\newcommand{\Sm}{{\bf S}}

\newcommand{\Km}{{\bf K}}

\newcommand{\Tm}{{\bf T}}

\newcommand{\cN}{{\cal N}}
\newcommand{\cM}{{\cal M}}
\newcommand{\cV}{{\cal V}}
\newcommand{\cQ}{{\cal Q}}

\newcommand{\cF}{{\cal F}}

\usepackage{lipsum}
\usepackage{amsfonts}
\usepackage{graphicx}
\usepackage{epstopdf}
\usepackage{algorithmic}
\ifpdf
  \DeclareGraphicsExtensions{.eps,.pdf,.png,.jpg}
\else
  \DeclareGraphicsExtensions{.eps}
\fi

\newsiamremark{remark}{Remark}
\newsiamremark{hypothesis}{Hypothesis}
\crefname{hypothesis}{Hypothesis}{Hypotheses}
\newsiamthm{claim}{Claim}

\headers{An iterative GKB algorithm in structural mechanics}{M. Arioli, C. Kruse, U. R\"ude and N. Tardieu}

\title{An iterative generalized Golub-Kahan algorithm for problems in structural mechanics }

\author{Mario Arioli \thanks{Libera Universita Mediterranea, Casamassima, Bari, Italy 
  (\email{arioli@lum.it}).}
\and Carola Kruse \thanks{Cerfacs, 29 Avenue Gaspard Coriolis, 31100 Toulouse, France
  (\email{carola.kruse@cerfacs.fr}, \email{ulrich.ruede@fau.de}).}
\and Ulrich R\"ude \footnotemark[3] \thanks{Friedrich-Alexander-Universit\"at Erlangen-Nuremberg, Cauerstr. 6, 91058 Erlangen, Germany (\email{ulrich.ruede@fau.de}). }  
\and Nicolas Tardieu \thanks{EDF R\&D, 7 Boulevard Gaspard Monge, 91120 Palaiseau, France
  (\email{nicolas.tardieu@edf.fr}).}}

\usepackage{amsopn}

\begin{document}

\maketitle

\begin{abstract}
This paper studies the Craig variant of the Golub-Kahan bidiagonalization algorithm
as an iterative solver for linear systems with saddle point structure. Such symmetric indefinite systems in 2x2 block form
arise in many applications, but standard iterative solvers are often found
to perform poorly on them and robust preconditioners may not be available. Specifically, such systems arise in structural mechanics, when a
semidefinite finite element stiffness matrix is augmented with linear multi-point constraints via Lagrange multipliers. Engineers often use such multi-point constraints to introduce boundary or coupling conditions into complex finite element models. The article will present a systematic convergence study of the Golub-Kahan algorithm
for a sequence of test problems of increasing complexity, including concrete structures enforced with pretension cables and the coupled finite element model of a reactor containment building.
When the systems are suitably transformed using
augmented Lagrangians on the semidefinite block and 
when the constraint equations are properly scaled, the
Golub-Kahan algorithm is found to exhibit excellent convergence that depends only
weakly on the size of the model.
The new algorithm is found to be robust in practical cases that are otherwise
considered to be difficult for iterative solvers. 
\end{abstract}



\section{Introduction}\label{sec:intro}

In structural mechanics, it is very common to impose kinematic relationships between degrees of freedom (DOF) in a finite element model. Rigid body conditions of a stiff part of a mechanical system or cyclic periodicity conditions on a mesh representing only a section of a periodic structure are typical examples of this approach. Such conditions can also be used to glue non-conforming meshes or meshes containing different types of finite elements. For example, we could link a thin structure modeled by shell finite elements to a massive structure modeled with continuum finite elements. These kinematic relationships are often called multi-point constraints (MPC) in standard finite element software and can be linear or nonlinear. In the case of a well-posed mechanical problem discretized with finite elements, the solution of the linearized problem can be expressed as the following constrained minimization problem
\begin{align}
\min_{\Am^T\wv = \rv} \frac{1}{2} \wv^T \Wm \wv - \gvv^T\wv, \label{eqn:minu}
\end{align}
where 
\begin{itemize}
\item[] \makebox[1.7cm]{\makebox[0.3cm]{$\Wm $\hfill} $\in \mathbb{R}^{m \times m}$\hfill} is the tangent stiffness matrix,
\item[] \makebox[1.7cm]{\makebox[0.3cm]{$\Am $\hfill} $\in \mathbb{R}^{m \times n}$\hfill} is the linearized matrix of the constraints,
\item[] \makebox[1.7cm]{\makebox[0.3cm]{$\wv$\hfill} $\in \mathbb{R}^{m}$\hfill} is the vector of nodal displacement unknowns,
\item[] \makebox[1.7cm]{\makebox[0.3cm]{$\gvv$\hfill} $\in \mathbb{R}^{m}$\hfill} is the volume force vector,
\item[] \makebox[1.7cm]{\makebox[0.3cm]{$\rv$\hfill} $\in \mathbb{R}^{n}$\hfill} is the data vector for inhomogeneous constraints.
\end{itemize}
With the introduction of Lagrange multipliers $\pv$, the augmented system that gives the optimality conditions for (\ref{eqn:minu}) reads
\begin{align}
\left[ \begin{array}{cc}
\Wm& \Am\\
\Am^T & 0
\end{array}\right]
\left[ \begin{array}{c}
\wv\\\pv
\end{array}\right]  =
\left[ \begin{array}{c}
\gvv\\ \rv
\end{array}\right]. \label{eqn:augsys}
\end{align}
In this article we assume that $\Wm$ is symmetric positive semidefinite, as it is typically the case when $\Wm$ arises from finite element models in structural mechanics. We additionally assume that
\begin{align}
\ker(\Wm)\cap \ker(\Am^T) = \left\{0\right\} \mbox{ and } \ker \Am = \left\{ 0\right\}. \label{eqn:WKerAt}
\end{align}
 To obtain a positive definite (1,1)-block in \cref{eqn:augsys}, a common method is to apply an \textit{augmented Lagrangian approach} as described by Golub and Greiff \cite{GoGr2003}. Let therefore $\Nm \in \mathbb{R}^{n\times n}$ be a positive symmetric definite matrix. Then we modify the leading block into


\begin{align}
\Mm:= \Wm + \Am \Nm^{-1} \Am^{T}.  \label{eqn:regular}
\end{align}
With the transformation 
\begin{align}
\begin{array}{lll}
\Mm &= &\Wm + \Am\Nm^{-1}\Am^T\\
\uv &= &\wv - \Mm^{-1}(\gvv - \Am\Nm^{-1}\rv)\\
\bv & = &\rv - \Am^T\Mm^{-1}(\gvv - \Am \Nm^{-1}\rv),
\end{array}
\label{eqn:trafo_semi_def}
\end{align}
\cref{eqn:augsys} is transformed into the equivalent system
\begin{align}
\left[
\begin{array}{cc}
\Wm + \Am \Nm^{-1} \Am^T & \Am\\
\Am^T & 0
\end{array}
\right]
\left[
\begin{array}{c}
\uv \\
\pv
\end{array}
\right]
=
\left[
\begin{array}{c}
0 \\
\bv
\end{array}
\right].\label{eqn:augsys_auglag}
\end{align}
This kind of regularization of the $(1,1)$-block is a common technique \cite{GoGr2003, BeGoLi2005, Ar2013}. It can also be applied when $\Wm$ is positive definite, with the goal that for a suitably chosen $\Nm$, we may find that \cref{eqn:augsys_auglag} becomes easier to solve than the original system. In the following, we will use the notation $\Mm$ for a positive definite matrix.\\
\\
The efficient solution of the above saddle point linear system \cref{eqn:augsys_auglag} has stimulated intensive research. One possible approach is to introduce the constraints on the continuous level, i.e. in the weak form of a PDE as with the mortar approach \cite{Bernardi1989}. In industrial software, when multi-point constraints are used, the constraints are however imposed on the already discretized equations. As it is furthermore usually not possible to make major modifications to an existing legacy code, any method of mortar-type becomes unfeasible. In this article, we will focus on the situation that the constraints are introduced on the discrete level, for which the solution of \cref{eqn:minu} remains a difficult task. We refer the reader to \cite{BeGoLi2005} for a comprehensive review of the topic. One of the commonly used methods is the Schur complement reduction technique, which requires an invertible (1,1)-block $\Mm$. It then has the advantage of solving two linear systems of size $m$ and $n$, instead of one system of size $m+n$. There is however the disadvantage that the Schur complement matrix $\Sm = -\Am^T \Mm^{-1} \Am $ may be dense and thus becomes expensive to solve. Krylov subspace methods for \cref{eqn:augsys_auglag} are reviewed in \cite{Saad2003}. In realistic finite element applications the saddle point matrix can be very poorly conditioned. As it is discussed in  \cite[section 3.5]{BeGoLi2005}, when the mesh size parameter $h$ goes to zero, the condition number of \cref{eqn:augsys_auglag} may increase. Krylov subspace methods will thus perform poorly with increasing problem size and rely on good preconditioning techniques. 
Another method to solve the saddle point system is based on an elimination technique \cite{Abel1979,jendele2009}. This strategy implies major modifications of the matrix of the linear system, whose profile can become much denser. Furthermore, the underlying algorithm is often sequential, where each constraint is treated one after the other. Consequently, this technique can not be used easily in a parallel framework. 
A different approach is used in \cite{stgeorges98}. The authors introduce a projector on the orthogonal of the kernel of the constraints matrix $\Am$ and solve the linear system on that subspace with an iterative method. This subspace projection technique is elegant and favorable convergence properties are shown. Unfortunately, the definition of the projector involves the factorization of the operator $\Am^T \Am$, which, in many practical cases, can be quite dense, causing the factorization to be expensive in time and space. Furthermore, one forward-backward substitution is needed at each iteration of the iterative method.\\
\\
In this paper we will focus on an iterative method for \cref{eqn:augsys_auglag} based on the Golub-Kahan bidiagonalization technique. We will find the iterates $\uv^k$ and $\pv^k$ separately, which requires to solve linear systems for $\Mm$ and for $\Nm$. We will show that for an appropriate choice of the matrix $\Nm$, the number of iterations required for convergence stays small and constant when the problem size increases. In particular, we will use this algorithm to solve problems in solid mechanics for which commonly used iterative solvers show a poor performance. Our test problems are generated by the finite element software code\_aster (\url{www.code-aster.org}). Code\_aster covers a wide range of physics including solid mechanics, thermics, acoustics, coupled thermo-hydro-mechanics and is also developed to numerically simulate critical industrial applications. It can treat steady-state and transient problems with various nonlinearities including frictional contact or complex constitutive laws. Code\_aster is developed since 1989 by one of the biggest electric utility companies in the world called EDF and is released as an open source software under GPL license since 2001. It is developed under Quality Insurance and it has been approved by the French (Autorit\'{e} de S\^{u}et\'{e} Nucl\'{e}aire) and English (Health and Safety Executive) Nuclear Regulatory Authorities to run numerical studies related to Nuclear Safety. The paper is organized as follows: We first introduce and review the Golub-Kahan bidiagonalization algorithm in \cref{sec:GGKB}. In \cref{sec:numexp}, we focus on models in structural mechanics and present a systematic convergence study. In \cref{sec:containment}, we will apply the proposed algorithm to a realistic industrial test case of a reactor containment building.

\section{The generalized Golub-Kahan bidiagonalization method}\label{sec:GGKB}

We will \, start by summarizing the main results of \cite{Ar2013} which are needed in our further discussion.
\subsection{Fundamentals of the Golub-Kahan bidiagonalization algorithm}
In the following, we will use the Hilbert spaces
\begin{align*}
 \mathcal{M} = \{{\bf v} \in \mathbb{R}^m: \|{\bf v} \|_{\Mm}^2 = {\bf v}^T \Mm {\bf v}\}, \hspace{0.3cm} \mathcal{N} = \{\qv \in \mathbb{R}^n: \|\qv \|_{\Nm}^2 = \qv^T \Nm \qv \}  
\end{align*}
and their dual spaces
\begin{align*}
 \mathcal{M}' = \{{\bf v} \in \mathbb{R}^m: \|{\bf v} \|_{\Mm^{-1}}^2 = {\bf v}^T \Mm^{-1} {\bf v}\}, \hspace{0.3cm} \mathcal{N}' = \{\qv \in \mathbb{R}^n: \|\qv \|_{\Nm^{-1}}^2 = \qv^T \Nm^{-1} \qv \}.  
\end{align*}
The scalar products for $\mathcal{M}$ and $\mathcal{N}$ are denoted by
\begin{align*}
 (\vvv_1,\, \vvv_2)_{\Mm} &= \vvv_1^T \Mm \vvv_2, &\forall \vvv_1, \vvv_2\in \mathcal{M},\\
  (\qv_1,\, \qv_2)_{\Nm} &= \vvv_1^T \Nm \qv_2, &\forall \qv_1, \qv_2\in \mathcal{N}.
\end{align*}
The respective scalar products in the dual spaces are given by
\begin{align*}
 (\vvv_1,\, \vvv_2)_{\Mm^{-1}} &= \vvv_1^T \Mm^{-1} \vvv_2, &\forall \vvv_1, \vvv_2\in \mathcal{M},\\
  (\qv_1,\, \qv_2)_{\Nm^{-1}} &= \vvv_1^T \Nm^{-1} \qv_2, &\forall \qv_1, \qv_2\in \mathcal{N}.
\end{align*}
Given $\qv \in \cM$ and $\vvv \in \cN$, we define the functional 
\begin{align}
 \cF : \cM\times \cN \rightarrow \mathbb{R}, \hspace{0.5cm} (q,v)\mapsto \dfrac{\vvv^T \Am \qv}{\|\qv\|_\Nm \; \|\vvv\|_\Mm}.\label{func}
\end{align}
The critical points of $\cF$ are the {\it elliptic singular values} and $\qv_i$,$\vvv_i$ are the {\it elliptic singular vectors} of $\Am$.
Indeed the saddle-point conditions for \cref{func} are
\begin{align}\label{GSVD}
\left\lbrace
\begin{array}{lcll@{}l}
\Am \qv_i &=& \sigma_i \Mm \vvv_i &\qquad \vvv_i^T \Mm \vvv_j &= \delta_{ij} \\
\Am^T \vvv_i &=& \sigma_i \Nm \qv_i  &\qquad \qv_i^T \Nm \qv_j &= \delta_{ij}
\end{array}
\right..
\end{align}
Hereafter, we assume that $\sigma_1 \ge \sigma_2 \ge \dots \geq \sigma_n > 0$.
If we operate a change of variables using $\Mm^{-\frac{1}{2}}$ and $\Nm^{-\frac{1}{2}}$,
 \begin{align}
 \left\{
  \begin{array}{l}
   \vvv = \Mm^{-1/2} x\\
   \qv = \Nm^{-1/2} y\\
  \end{array}
 \right.
 \end{align}
we have that the elliptic singular values are the standard singular values of
$$\tilde{\Am} = \Mm^{-1/2} \Am \Nm^{-1/2}.$$
The generalized singular vectors
 $\qv_i$ and $ \vvv_i$, $i = 1, \dots ,n$  are  the transformation by
$\Mm^{-1/2}$ and $\Nm^{-1/2}$ respectively of the left and right standard singular vector of $\tilde{\Am}$ \cite{Ar2013}.\\
\\
In \cite{GoKa1965, PaSa1982}, several algorithms for the bidiagonalization of a $m \times n$ matrix are presented.
All of them can be theoretically applied to $\tilde{\Am}$ and their generalization to $\Am$ is straightforward as shown
by Benbow \cite{Benbow1999}.
Here, we  will specifically analyze one of the variants known as the "Craig"-variant \cite{PaSa1982, Sa1995, Sa1997}. We seek the matrices $\Qm \in \mathbb{R}^{n\times n}, \Vm \in \mathbb{R}^{m\times m}$ and the bidiagonal matrix $\Bm$, such that the following relations are satisfied
\begin{align}
\left\lbrace
\begin{array}{r@{}c@{}ll@{}l}
\Am \Qm &=& \Mm \Vm \left[ \begin{array}{c}\Bm\\ 0\end{array}  \right] &\qquad \Vm^T \Mm \Vm &= \Id_m \\
&&\\
\Am^T \Vm &=& \Nm \Qm \left[  \Bm^T  ; 0 \right] &\qquad \Qm^T \Nm \Qm &= \Id_n
\end{array}
\right. \label{eqn:GKalg}
\end{align}
where
\begin{eqnarray*}
\Bm =
\left[ \begin{array}{ccccc}
\alpha_1 & \beta_1 &  0 & \cdots & 0 \\
0 & \alpha_2 & \beta_2 &   \ddots & 0 \\
\vdots &\ddots  & \ddots  & \ddots &\ddots  \\
0 & \cdots & 0 &\alpha_{n-1} & \beta_{n-1}   \\
 0 & \cdots &  0 &  0 & \alpha_n
\end{array}\right]  .
\end{eqnarray*}
We apply the above relations to the augmented system
\begin{align}
\left[ \begin{array}{cc}
\Mm& \Am\\
\Am^T & 0
\end{array}\right]
\left[ \begin{array}{c}
\uv\\\pv
\end{array}\right]  =
\left[ \begin{array}{c}
0\\ \bv
\end{array}\right]. \label{eqn:augsys_GKB}
\end{align}
By the change of variables
\begin{align}
\left\lbrace
\begin{array}{l}
\uv = \Vm \hat{\zv} \\
\pv = \Qm \hat{\yv}
\end{array}\right. \label{eqn:chvar}
\end{align}
and by multiplying the system from the left by
\begin{align*}
\left[
\begin{array}{cc}
\Vm^T & 0\\
0 & \Qm^T
\end{array}
\right],
\end{align*}
the augmented system can be transformed with \cref{eqn:GKalg} into
\begin{eqnarray*}
\left[ \begin{array}{ccc}
\Id_n   & 0              & \Bm\\
0         & \Id_{m-n} & 0 \\
\Bm^T & 0              & 0
\end{array}\right]
\left[ \begin{array}{c}
\hat{\zv}_1 \\ \hat{\zv}_2 \\ \hat{\yv}
\end{array}\right]  =
\left[ \begin{array}{c}
0 \\ 0 \\ \Qm^T \bv
\end{array} \right] .
\end{eqnarray*}
We see that $\hat{\zv} = (\hat{\zv}_1, \hat{\zv}_2) = (\hat{\zv}_1, 0)$. Consequently, $\uv$ only depends on the first $n$ columns of $\Vm$  and thus the system reduces to
\begin{align*}
\left[ \begin{array}{cc}
\Id_n   & \Bm\\     \Bm^T       & 0
\end{array}\right]
\left[ \begin{array}{c}
\hat{\zv}_1  \\ \hat{\yv}
\end{array}\right]  =
\left[ \begin{array}{c}
 0 \\ \Qm^T \bv
\end{array} \right] .
\end{align*}

To define a bidiagonalization algorithm, we choose the first vector $\qv_1$ in $\Qm^T \Nm \Qm$ as
\begin{align*}
\qv_1 = \Nm^{-1} \bv / \|\bv\|_{\Nm^{-1}}.
\end{align*}
A straightforward calculation then shows that 
\begin{align*}
\Qm^T \bv = \ev_1 \|\bv\|_\Nm.
\end{align*}
In \cite{Ar2013}, it is proved that denoting by $\zeta_j$ the entries of $\hat{\zv}$,  taking 
advantage of the recursive properties of the Golub-Kahan algorithm \cite{GoKa1965}, and using some of
the results of \cite{PaSa1982}, we can obtain a fully  recursive algorithm.
The final Golub-Kahan bidiagonalization algorithm is presented in Algorithm \ref{alg:GKB}.
\begin{algorithm}[H]
  \caption{Craig's variant algorithm}
  \label{alg:GKB}
  \begin{algorithmic}
  \REQUIRE{$\Mm , \Am , \Nm, \bv$, maxit}
  \STATE{$\beta_1 = \|\bv\|_{\Nm^{-1}}$;  $\qv_1 = \Nm^{-1} \bv / \beta_1$}
  \STATE{$\wv = \Mm^{-1} \Am \qv_1$; $\alpha_1 = \|\wv\|_{\Mm}$; $\vvv_1 = \wv / \alpha_1$}
  \STATE{$\zeta_1 = \beta_1 / \alpha_1$; $\dv_1=\qv_1/ \alpha_1$; $\pv^{(1)} = - \zeta_1 \dv_1$}
  \WHILE{convergence = false  and $k < $ maxit}
  \STATE{$k = k + 1$}
  \STATE{$\gvv = \Nm^{-1} \left( \Am^T \vvv_k - \alpha_k \Nm \qv_k  \right) $; $\beta_{k+1} = \|\gvv\|_{\Nm}$}
  \STATE{$\qv_{k+1} = \gvv / {\beta_{k+1}}$}
  \STATE{$\wv = \Mm^{-1} \left(  \Am \qv_{k+1} - \beta_{k+1} \Mm \vvv_{k} \right)$; $\alpha_{k+1} = \|\wv\|_{\Mm}$}
  \STATE{$\vvv_{k+1} = \wv / {\alpha_{k+1} }$}
  \STATE{$\zeta_{k+1} = - \dfrac{\beta_{k+1}}{\alpha_{k+1}} \zeta_k$}
  \STATE{$\dv_{k+1} = \left( \qv_{k+1} - \beta_{k+1} \dv_k \right) / \alpha_{k+1} $}
  \STATE{$\uv^{(k+1)} = \uv^{(k)} + \zeta_{k+1} \vvv_{k+1}$; $\pv^{(k+1)} = \pv^{(k)} - \zeta_{k+1} \dv_{k+1}$}
  \STATE{$\left[ \right.  $ convergence $\left. \right] $ = check$(\zv_k, \dots)$}
  \ENDWHILE
  \RETURN $\uv^{k+1}, \pv^{k+1}$
  \end{algorithmic}
\end{algorithm}
We highlight that, in the following, the values of $\zeta_k$, $\alpha_k$ and $\beta_k$ will be
always those as computed in \cref{alg:GKB}. Furthermore note that in each iteration two linear systems, one for $\Mm$ and one for $\Nm$ have to be solved. Furthermore, the Craig algorithm has an important property of minimization. Let $\cV = \Span\left\{\vvv_1,...,\vvv_k \right\}$ and $\cQ = \Span\left\{\qv_1,...,\qv_k \right\}$. At each step $k$, the \cref{alg:GKB} computes $\uv^{(k)}$ such that \cite{Sa1995}
\begin{align}
 \min_{\uv^{(k)}\in \cV, \,(\Am^T\uv^{(k)}-\bv)\perp \cQ} \|\uv - \uv^{(k)}\|_{\Mm}. \label{eqn:minprop}
\end{align}

\subsection{Convergence properties of the Golub-Kahan algorithm}\label{sec:genprob}
We now consider an augmented system with a positive definite (1,1)-block $\Wm$. We apply the augmented Lagrangian approach $\Mm = \Wm + \Am \Nm^{-1} \Am^{T}$ of \cref{eqn:regular}, where the matrix $\Nm$ corresponds to the one in \cref{eqn:GKalg}. With the transformation \cref{eqn:trafo_semi_def}, we arrive at an augmented system of the form \cref{eqn:augsys_auglag}. We follow the discussion in \cite{GoGr2003} and choose
\begin{align*}
\Nm = \frac{1}{\eta} \Id.
\end{align*}
For an appropriate choice of $\eta$, the following theorem states our main result on the convergence of the GKB method.
\begin{theorem}\label{thm:eta}
Let $\Mm = \Wm + \eta \Am \Am^T$ and $\Wm$ be positive definite matrices and $\lambda_1 \leq \dots \leq \lambda_n$ be the eigenvalues of $\Am^T \Wm^{-1} \Am$. \\
\begin{center}
If $\eta \geq \lambda_1^{-1} > 0$, then $\kappa(\tilde{\Am}) \leq \sqrt{2}$
\end{center}
\end{theorem}

\begin{proof}
Let 
\begin{align*}
\sigma_1 \leq \dots \leq \sigma_n 
\end{align*}
be the elliptic singular values of $\Am$ with $\Mm$ and $\Nm$ norms as in \cref{GSVD}.
From \cref{GSVD} follows
\begin{align*}
 \eta \Am^T \Mm^{-1} \Am p_i = \sigma_i^2 p_i.
\end{align*}
Thus $\mu_i = \sigma^2_i$ are the eigenvalues of
\begin{align*}
\eta \Am^T \bigl(\Wm + \eta \Am \Am^T \bigr)^{-1} \Am.
\end{align*}
With the Sherman-Morrison formula, we obtain
\begin{align*}
 \eta \Am^T \bigl(\Wm + \eta \Am \Am^T \bigr)^{-1} \Am 
 = \eta \Am^T \Wm^{-1}\Am \bigl(\Id + \eta \Am^T \Wm^{-1} \Am \bigr)^{-1} 
\end{align*}
Let $\lambda_1 \leq \dots \leq \lambda_n$ be the eigenvalues of $\Am^T \Wm^{-1} \Am$. Then
\begin{align*}
\mu_i = \dfrac{\eta \lambda_i}{1 + \eta\lambda_i} \qquad \forall i.
\end{align*}
We obtain for the condition number of $\tilde{\Am} = \Mm^{-\frac{1}{2}}\Am \Nm^{-\frac{1}{2}}=\eta\bigl(\Wm + \eta \Am \Am^T \bigr)^{-\frac{1}{2}}\Am $
\begin{align*}
\kappa^2(\tilde{\Am}) = \dfrac{\mu_{max}}{\mu_{min}} \leq \dfrac{1+\eta \lambda_1}{\eta \lambda_1}
\end{align*}
It follows that if $\eta \geq \lambda_1^{-1}$, then $\kappa(\tilde{\Am}) \leq \sqrt{2}$.
\end{proof}
From the previous result, we can conclude that if we choose $\eta$ big enough, the condition number of $\tilde{\Am}$ is bounded by $\sqrt{2}$. In \cite[Section 4.2]{OrAr2017}, it is discussed that the standard Golub-Kahan bidiagonalization process applied to $\tilde{\Am} = \Mm^{-1/2}\Am \Nm^{-1/2}$ is equivalent to the generalized Golub-Kahan bidiagonalization applied to $\Am$. We can thus conclude from \cref{thm:eta}, that \cref{alg:GKB} exhibits excellent convergence properties and that only few iterations should be necessary to obtain sufficiently accurate results. As second desirable property, we can expect the number of iterations to be independent of the mesh size for problems coming from constrained FEM discretizations, as long as we choose $\eta$ big enough.
\\
\\
However, there is no such thing as a free lunch. In each iteration in \cref{alg:GKB}, we have to solve linear systems with the matrices $\Mm$ and $\Nm$. While  $\Nm^{-1} = \eta \Id$ is trivial, the condition number of $\Mm$ depends on $\eta$ and thus on the smallest eigenvalue of
$\Am^T \Wm^{-1} \Am$. The condition number of the resulting matrix $\Mm = \Wm + \eta \Am \Am^T$ could become very large for large $\eta$. The solution of the linear systems in \cref{alg:GKB} may thus become difficult, and additional numerical errors may be introduced. The possibly high condition number of $\Mm$ is especially problematic for large scale problems, when an inner direct solver is no longer applicable and an iterative solver is applied. It is thus crucial to find an optimal balance of $\eta$ to enable an efficient inner solution step. The numerical experiments suggest that in practice reasonable values of $\eta$ proportional to
$|| \Wm ||_1$ reduce $\kappa(\tilde{\Am})$ sensibly, without dramatically increasing the ill-conditioning of $\Mm$.

\subsection{Stopping criteria} \label{sec:stopcrit}
In the following, we summarize possible stopping criteria for the GK bidiagonalization algorithm as suggested in \cite{Ar2013}.
\subsubsection{A lower bound estimate}
First, we look at a lower bound estimate of the error in the energy norm. The error $\ev^{(k)} = \uv - \uv^{(k)}$ can be expressed using the \Mm-orthogonality property of $\Vm$ and \cref{eqn:chvar} by
\begin{align*}
\| \ev^{(k)} \|_{\Mm}^2 = \sum_{j=k+1}^n \zeta_j^2 = \Big|\Big| \hat{\zv} - \left[ \begin{array}{c}
\zv_k \\ 0
\end{array}\right] \Big|\Big|_2^2.
\end{align*} 
To compute the error $\ev^{(k)}$, we thus need $\zeta_{k+1}$ to $\zeta_{n}$, which are available only after the full $n$ iterations of the algorithm. Given a threshold $\tau < 1$ and an integer $d$, we can define a lower bound of $\| \ev^{(k)} \|_{\Mm}^2$ by
\begin{align} 
\xi_{k,d}^2 = \sum_{j=k+1}^{k+d+1} \zeta_j^2 < \| \ev^{(k)} \|_{\Mm}^2 .
\end{align}
$\xi_{k,d}$ measures the error at step $k-d$, but as the following $\uv^{(k)}$ minimize the error due to \cref{eqn:minprop}, we can safely use the last ones. Also, this lower bound estimate is very inexpensive to compute and it has additionally the advantage that it yields an upper bound for the residual in the dual norm defined by $\Nm^{-1}$
\begin{align*}   
\| \Am^T \uv^{(k)} - \bv \|_{\Nm^{-1}} =  | \beta_{k+1} \; \zeta_k |  \le \sigma_1  | \zeta_k | 
= \| \tilde{\Am} \|_2  | \zeta_k |< \| \tilde{\Am} \|_2\tau.
\end{align*}
With a carefully chosen $d$,  procedure ``{check$(\zv_k, \dots) $}'' in \cref{alg:GKB} can then be constructed as \cref{alg:lowbound}.
\begin{algorithm}[H]
\caption{Lower bound estimate}
\label{alg:lowbound}
\begin{algorithmic}
\REQUIRE $\zv_k , k , k, d, \tau$
\STATE convergence = false;
\IF {$k > d$}
\STATE $\xi^2 =\sum_{j=k-d +1}^{k} \zeta_j^2$;
\IF {$\xi \leq \tau$}
\STATE convergence = true;
\ENDIF
\ENDIF
\RETURN convergence
\end{algorithmic}
\end{algorithm}

\subsubsection{An upper bound estimate}
To define a stopping criterion for the GKB method, it is useful to also have an upper bound error estimate. Obviously, this estimate is more reliable than the previous lower bound. The following approach has been presented in \cite{Ar2013}. It is inspired by
the Gauss-Radau quadrature algorithm and similar to the one described in \cite{GoMeu2010}. Let therefore $\Tm = \Bm^T \Bm$. $\Tm$ is a non-negative, triagonal and positive definite matrix of entries

\begin{align*}
\left\{
\begin{array}{ll}
\Tm_{1,1} = \alpha_1^2, & \\
\Tm_{i,i} = \alpha_i^2 + \beta_i^2, & i = 2,..,n, \\
\Tm_{i,i+1} = \Tm_{i+1,i} = \alpha_i\beta_{i+1}, & i = 1,..,n, \\
0 & \mbox{otherwise}.
\end{array}
\right.
\end{align*}
With straightforward calculations, we have
\begin{align*}
 \| \ev^{(k)} \|_{\Mm}^2 = \sum_{j=k+1}^n \zeta_j^2 = \|b\|_{\Nm}^2  \left[  \left( \Tm^{-1}\right) _{1,1} - \left( \Tm_k^{-1}\right) _{1,1} \right],
\end{align*}
where $\Tm_k$ is the $k \times k$ principal submatrix of $\Tm$ \cite{GoMeu2010}. Let $0 < a < \sigma_n $   a lower bound for all the singular values of $\Bm$.
We compute the matrix $\hat{\Tm}_{k+1}$ as
\begin{eqnarray*}
\hat{\Tm}_{k+1} = \left[
\begin{array}{cc}
\Tm_k & \alpha_k \beta_k \ev_k\\
\alpha_k \beta_k \ev_k^T& \omega_{k+1}
\end{array}
\right] ,
\end{eqnarray*}
where $\omega_{k+1} = a^2 + \delta_k(a^2) $  and $\delta_k(a^2) $ is the $k$-entry of the solution of
\begin{align*}
\left( \Tm_k - a^2 \Id \right) \mathbf{\delta}(a^2)  = \alpha_k^2 \beta_k^2 \ev_k .
\end{align*}
We point out that the matrix $(\Tm_k-a^2\Id)$ is positive definite and that $\hat{\Tm}_{k+1}$ has one eigenvalue equal to $a^2$. Analogously to what is done in \cite{GoMeu2010} for the conjugate gradient method, we can recursively compute
$\delta(a^2)_k$ and $\omega_{k+1}$ by using the Cholesky decomposition. The pseudo-code for obtaining the upper bound estimate $\Xi$ is presented in \cref{alg:upbound}. It is a practical realization of a Gauss-Radau quadrature that uses the matrices $\hat{\Tm}_k$. Therefore, from \cite[Theorem 6.4]{GoMeu2010}, we can derive that $\Xi$ is an upper bound for $\|\ev^{(k)}\|_{\Mm}$. Although this upper bound estimate gives a reliable stopping criterion, its calculation is in practice very difficult to obtain owing to the need of an accurate estimate of the smallest singular value. In the following numerical experiments, we will use exclusively the lower bound stopping criterion. For any further details on error estimates and global bounds, we refer to \cite{Ar2013}.

\begin{algorithm}
\caption{checkUB}
\label{alg:upbound}
\begin{algorithmic}
\REQUIRE{$\zv_k, k, d, \tau, a,  \|b\|_{\Nm}, \Bm_k $}
\STATE{convergence = false;}
\IF {$k = 1$}
\STATE $\bar{d}_1 = \alpha_1^2 + \beta_1^2 - a^2$;
\ELSE
\STATE $\bar{d}_k = \alpha_k^2 + \beta_k^2 - \varpi_{k-1}$;
\ENDIF
\STATE $\varpi_k = a^2 + \dfrac{\alpha_k^2 \beta_k^2}{\bar{d}_k}$; $\;\;\varphi_k = \dfrac{\beta_k^2 \zeta_k^2}{\sqrt{ \bar{d}_k +a^2 - \beta_k^2}}$ ;
\IF {$k > d$}
\STATE $\xi^2 =\sum_{j=k-d +1}^{k} \zeta_j^2$;  $\qquad \Xi^2 = \xi^2 + \varphi_k$;
\IF {$\Xi \le \tau$}
\STATE convergence = true;
\ENDIF
\ENDIF
\RETURN convergence
\end{algorithmic}
\end{algorithm}

\section{Numerical Experiments} \label{sec:numexp}
In the following, we will apply the generalized GKB method to augmented matrix systems generated in the open source all-purpose finite element software code\_aster. In each test case, the models obey the laws of linear elasticity. We focus on the equilibrium of an elastic body under the small displacement hypothesis, for which the problem is to find the displacement field $\uv$ with $\uv:\bar{\Omega}\rightarrow\mathbb{R}^3$ such that
\begin{align}
 -\Div(\sigma(\uv)) &= \fv, &\mbox{ in } \Omega, \label{eqn:elas1}\\
 \sigma(\uv)n &= \hv, &\mbox{ on } \Gamma_N,\\
 \uv &= \uv_D, &\mbox{ on } \Gamma_D.
\end{align}
Here $\hv$ and $\uv_D$ are the Neumann and the Dirichlet data and the stress and strain tensors are defined as
\begin{align}
 \sigma(\uv)&= C\epsilon(\uv),\\
 \epsilon(\uv) &= (\nabla \uv +\nabla^T \uv)/2. \label{eqn:constlaw}
\end{align}
In the elastic case, $C$ is the fourth order elastic coefficient (or Hooke's law) tensor satisfying both symmetry and ellipticity conditions. Furthermore, the constitutive law (\ref{eqn:constlaw}) connects linearly $\sigma$ to the strain tensor field $\epsilon$.  Although we know the underlying physical model of the test cases, the following convergence analysis of the GKB algorithm is done only on matrix level. We thus refer the interested reader for any further details on the finite element discretization of \cref{eqn:elas1} to \cref{eqn:constlaw} used in code\_aster to \cite{Abbas2013}.\\
\\
The simulations in this section are done in Matlab. We will use the Matlab backslash solver for the inversion of $\Mm$ and $\Nm$ in \cref{alg:GKB}.

\subsection{Example: Cylinder} \label{ex1:cylinder}
As our first example, the domain $\Omega$ is chosen as a thick-walled cylinder as illustrated in \cref{fig:cyl}. The model is a classical linear elasticity system, as described above, with $m$ degrees of freedom approximated by a linear finite element method. Dirichlet boundary conditions are imposed on the left end and are shown in green. Furthermore, MPCs are applied to obtain a rigid inner ring, which is illustrated in \cref{fig:cyl} by the gray elements. For the derivation of the constraint equations, we refer to \cite{Pe2011}. These kinematic relationships ensure that the inner ring resists any kind of outer forces.
\\

\begin{figure}[htb]
\begin{center}
\includegraphics[width=9.0cm]{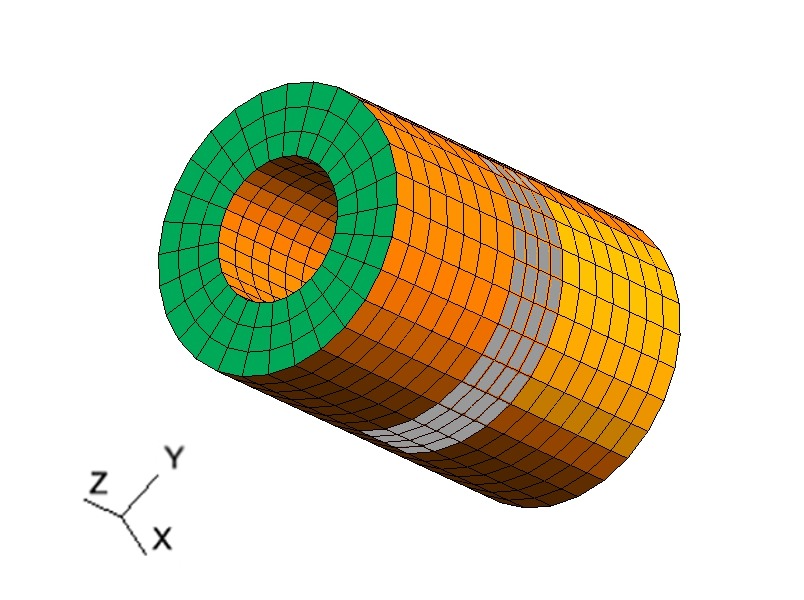}
\end{center}
\caption{Cylinder with rigid ring and Dirichlet boundary conditions.}
\label{fig:cyl}
\end{figure}

\subsubsection{Matrix setup} \label{sec:extr_data}
A double Lagrange multiplier approach \cite{Pe2011_01} is used in code\_aster  which leads to augmented systems with the structure

\begin{align}
\Km = \left(
\begin{array}{ccc}
\Wm & \gamma \Am & \gamma \Am\\
\gamma \Am^T & -\gamma I & \gamma I\\
\gamma \Am^T & \gamma I & -\gamma I
 \end{array}
\right).\label{eqn:doubleLg}
\end{align}
Here, $\Wm$ is the positive definite elasticity stiffness matrix, $\Am$ is the stiffness constraint matrix following the derivation in \cite{Pe2011} and $\gamma := \frac{1}{2}(\min\Wm_{ii}+\max\Wm_{ii})$ are multiplicative factors to equilibrate the scaling of the blocks.
After extraction of the matrices $\Wm$ and $\gamma \Am$, we thus get

 \begin{align}
 \left(
 \begin{array}{cc}
 \Wm & \gamma \Am \\
 \gamma \Am^T & 0
  \end{array}
 \right)
 \left(
 \begin{array}{c}
 \uv  \\
 \lambda
  \end{array}
 \right)
 =
 \left(
 \begin{array}{c}
 \gvv  \\
 0
  \end{array}
 \right). \label{eqn:extractedSystem}
\end{align}
The structure of the augmented system is shown in \cref{fig:mataugG}. Furthermore, we observe that the system \cref{eqn:extractedSystem} can be simplified by scaling it by $\gamma$. To exploit the result of \cref{thm:eta}, we modify the $(1,1)$-block as described in \cref{eqn:regular} and \cref{sec:genprob} to
\begin{align}
\Mm = \frac{1}{\gamma}\Wm + \eta \Am\Am^T \label{eqn:trafomat}
\end{align}
and transform \cref{eqn:extractedSystem} following \cref{eqn:trafo_semi_def} to obtain a system of type \cref{eqn:augsys_GKB}. The exact solutions are obtained by solving the original augmented system \cref{eqn:doubleLg} for a given right-hand side received from code\_aster, using the Matlab backslash solver. The delay parameter of \cref{alg:GKB} is chosen as $d=5$ and the tolerance as $\tau = 10^{-5}$.

\begin{figure}
\begin{center}
\includegraphics[width=9.0cm]{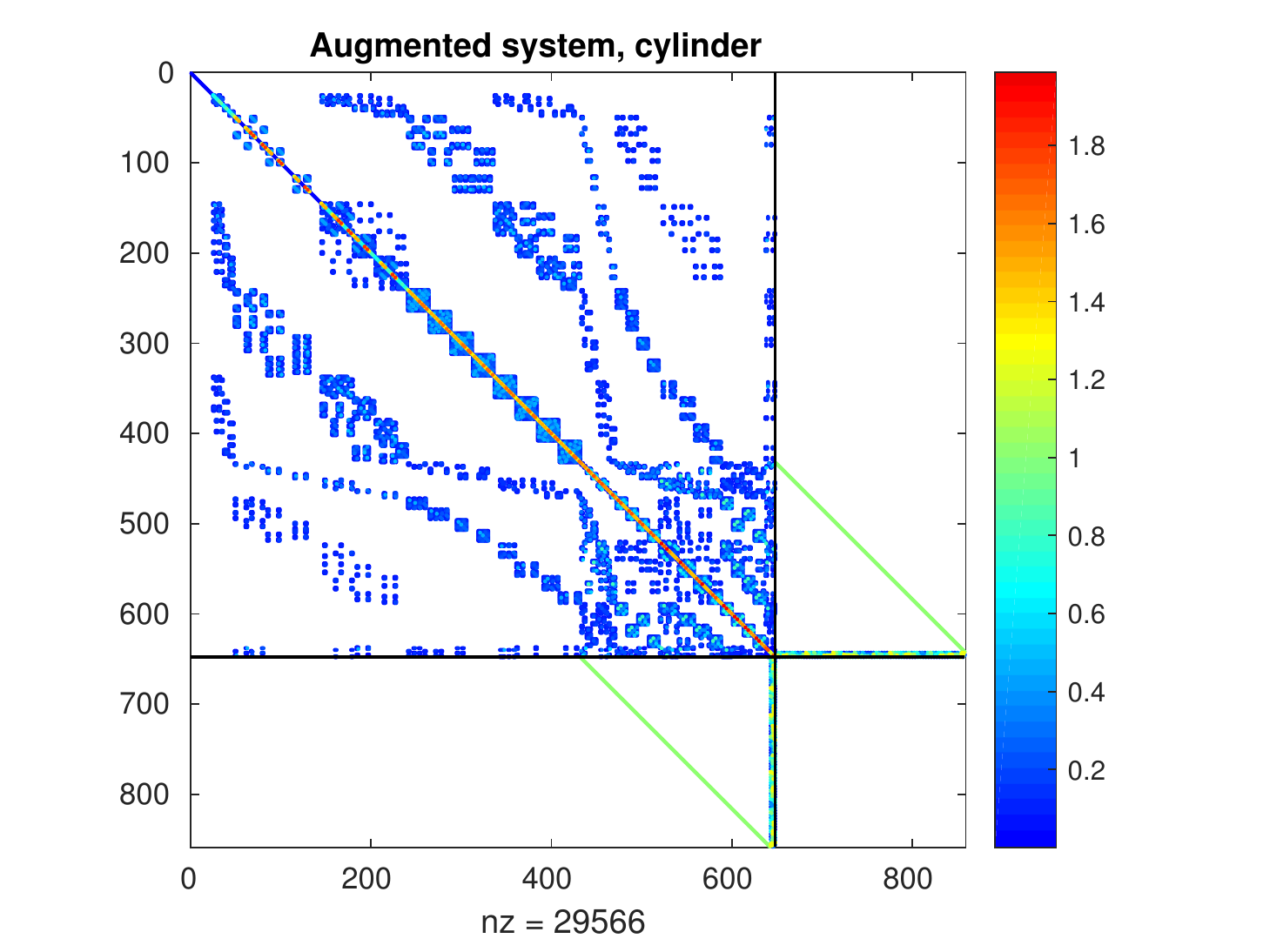}
\end{center}
\caption{Augmented matrix system for cylinder, Problem 1}
\label{fig:mataugG}
\end{figure}

\subsubsection{Results}
We define four test problems with increasing resolution. In \cref{tab:probsize}, the number of the degrees of freedoms can be found, where $m$ corresponds to the physical degrees of freedom, $n$ corresponds to the number of constraints and $nnz$ stands for the non-zero entries of the respective sparse matrices. We choose $\eta = \frac{1}{\gamma}\| \Wm \|_1$. The transformation \cref{eqn:trafomat} increases the number of nonzero entries, but the ratios still stay reasonably small. In \cref{tab:condnum1}, the condition numbers and norms of the occurring matrices are presented. The condition number of $\Mm$ does increase in $\eta$ (see \cref{sec:GGKB}).\\
\\
\begin{table}[htb]
\caption{Test problem sizes}
\begin{center}
\begin{tabular}{|c|c|c|c|c|c|}
\hline
name & $m$ & $n$ & $nnz(\Mm)$ & $nnz(\Am)$& $nnz(\Wm)$\\
\hline
Prob. 1       &648      &   210& 30080& 1259 & 28296 \\
Prob. 2 	  &2520    &   714& 147800& 4985 & 139636\\
Prob. 3 	  &6384 &      1674&  409246 &10045  & 392816\\
Prob. 4 & 46620 & 8814 & 3367462 & 26436& 3262086\\
\hline
\end{tabular}
\end{center}
\label{tab:probsize}
\end{table}

\begin{table}[htb]
\caption{Norms and condition numbers of matrices}
\begin{center}
\begin{tabular}{|c|c|c|c|c|}
\hline
name & $\eta = \frac{1}{\gamma}||\Wm||_1$& $\kappa(\Mm)$ &  $\kappa(\Wm)$& $||\Am||_1$\\
\hline
Prob. 1 & 9.13 &$ 8.3\cdot 10^5$ & $5.8 \cdot 10^3$& 6.27 \\
Prob. 2 & 8.95  & $7.1 \cdot 10^6$ & $1.9 \cdot 10^4$&5.79\\
Prob. 3 & 8.86  & $3.0 \cdot 10^{7}$ & $3.5 \cdot 10^4$&5.74\\
Prob. 4 &  8.96 & $5.0 \cdot 10^{8}$ & $1.2 \cdot 10^5$ & 5.34 \\
\hline
\end{tabular}
\end{center}
\label{tab:condnum1}
\end{table}
The convergence plots with upper and lower bound estimates of the GKB method are presented in \cref{fig:errUL12,fig:errUL34}. The error of the GKB solution obtains the required tolerance of $10^{-5}$ already after 6 iterations for the smallest problem and after 7, 8 and 9 for Problems 2 - 4 (see \cref{fig:errUL12,fig:errUL34}), respectively. The lower bound for the error at iteration $k$ is however computed only when iteration $k+d$ has been reached. Consequently, the GKB stops only after 11 to 14 iterations. This also explains why the final errors are remarkably smaller than the sought precision. We observe that although the number of DOF increases from Problems 1 to 4, the number of iterations increases by only 1 for each finer mesh and the algorithm stops after 14 iterations at most. To obtain a complete independence of the mesh size as it is shown in \cref{thm:eta}, $\eta$ would need to be chosen bigger. This will be discussed in the following section.

\begin{figure}
\begin{center}
\includegraphics[width=6.4cm]{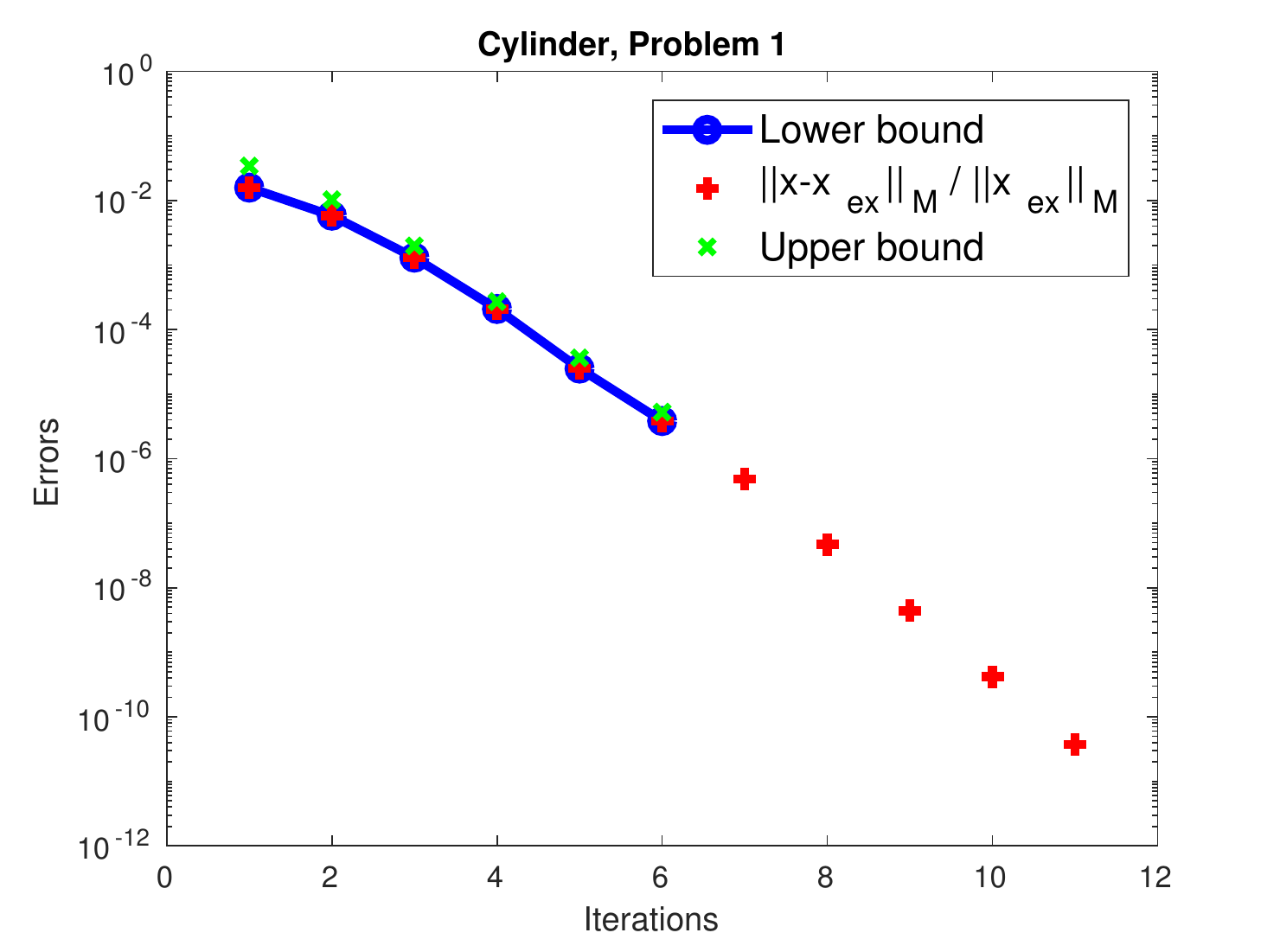}
\includegraphics[width=6.4cm]{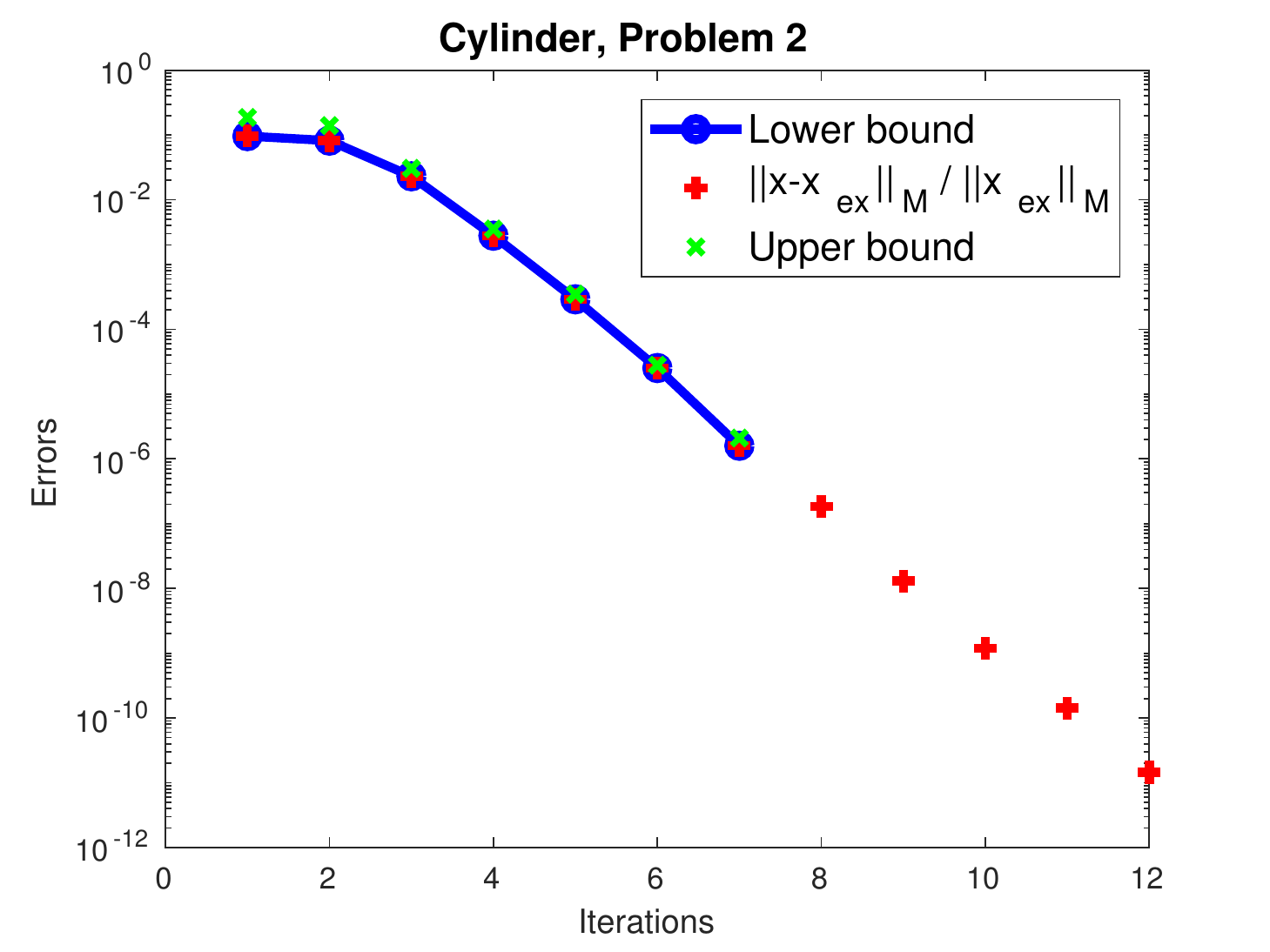}
\end{center}
\caption{Convergence of generalized GKB method for Problems 1 and 2.}
\label{fig:errUL12}
\end{figure}

\begin{figure}
\begin{center}
\includegraphics[width=6.4cm]{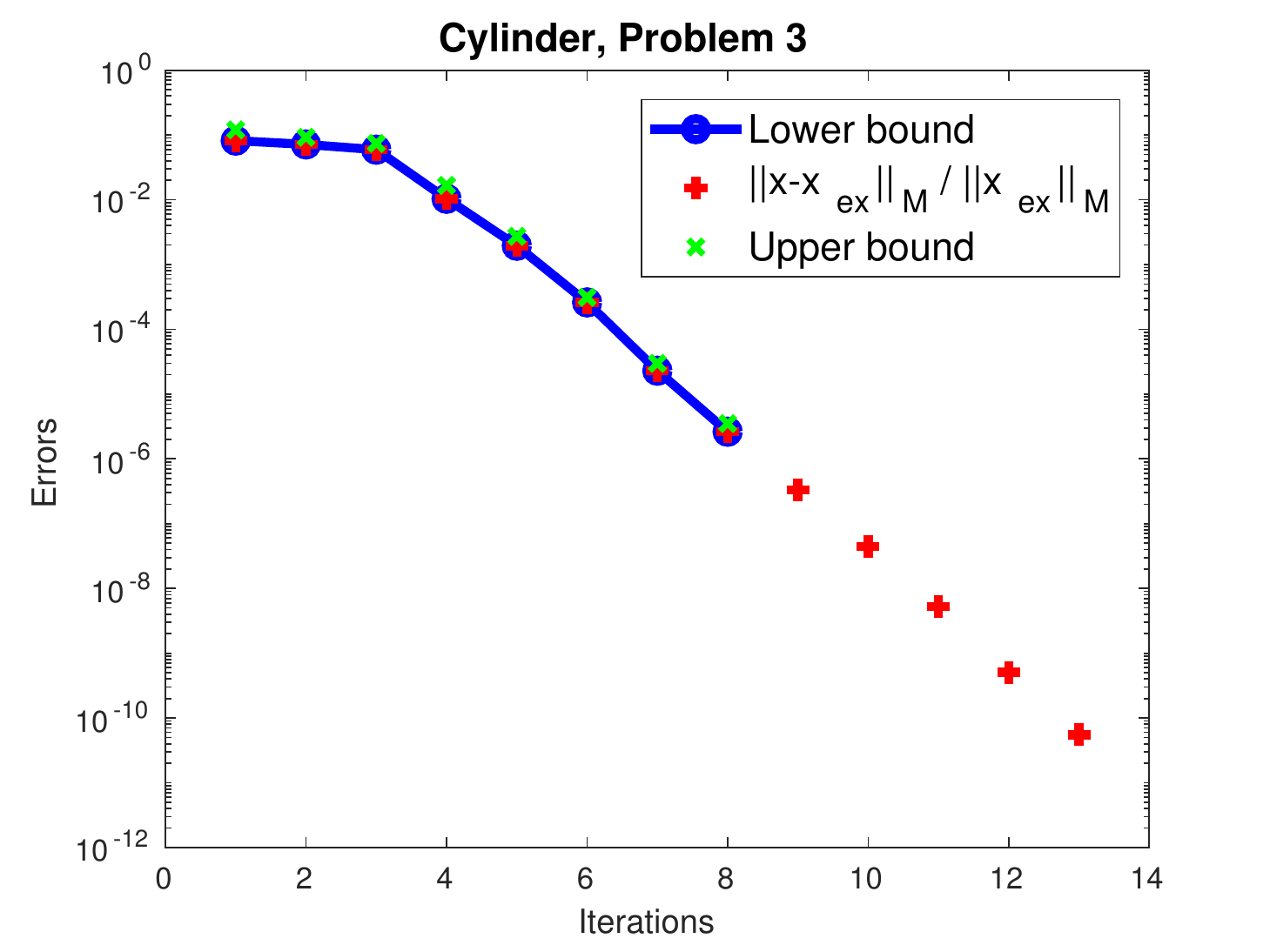}
\includegraphics[width=6.4cm]{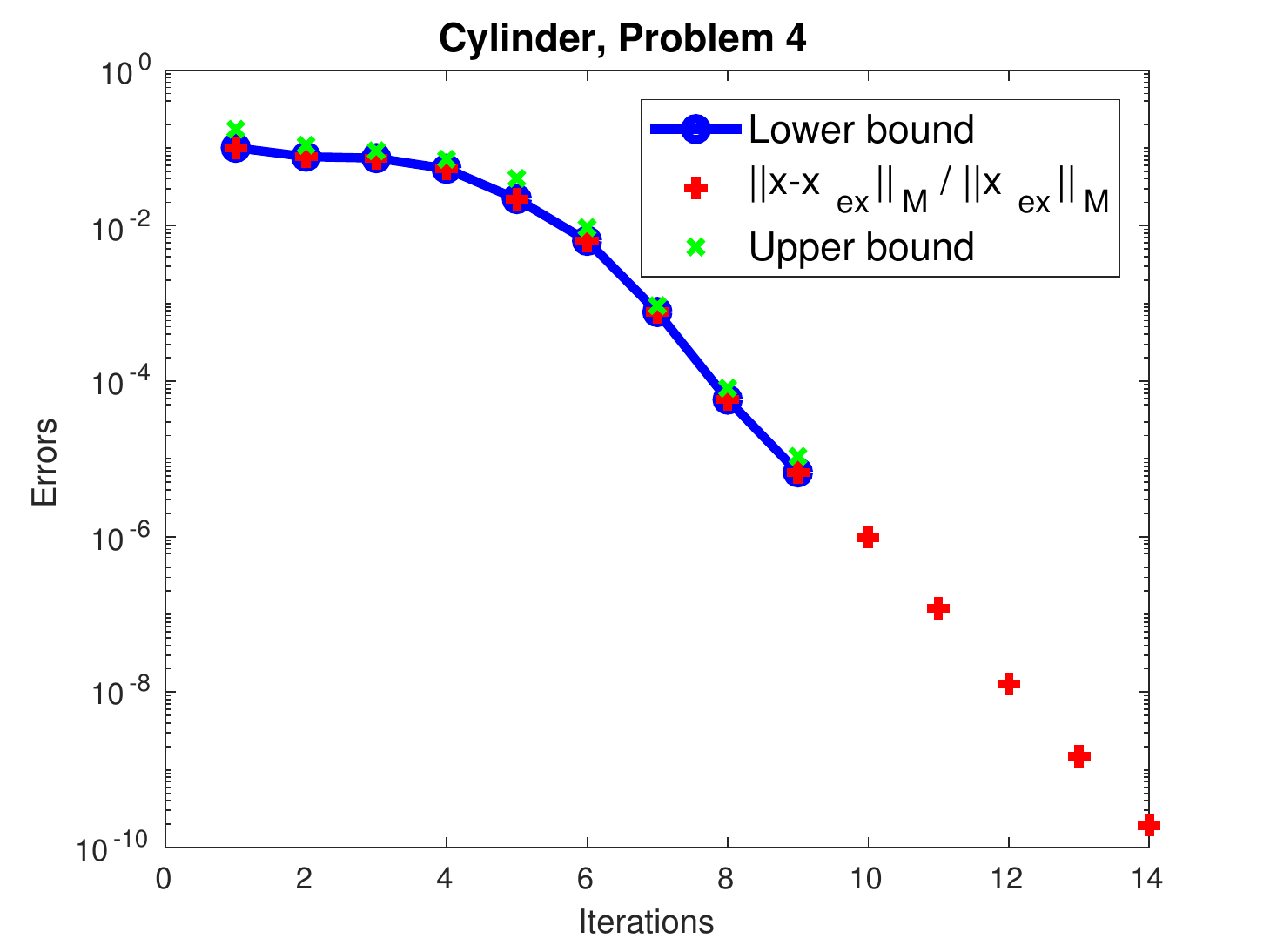}
\end{center}
\caption{Convergence of generalized GKB method for Problems 3 and 4.}
\label{fig:errUL34}
\end{figure}

\subsubsection{Choice of $\Nm$}
In the previous numerical examples, we choose the parameter $\eta = \frac{1}{\gamma}\|\Wm\|_1$ to better represent the energy subject to the MPC constraints, as described in the augmented system. 
The recommendation of Golub and Greiff in \cite{GoGr2003}, who found numerically that $\eta = \gamma \frac{\|\Wm\|}{\|\Am\|^2}$ could be a good value, leads to too small an $\eta$ for our practical examples. With this choice, we found that the number of iterations increases noticeably.
In \cref{thm:eta}, we proved that for $\eta \geq \lambda_1^{-1}$, the condition number of $\tilde{\Am}$ is bounded above and the number of iterations in \cref{alg:GKB} is independent of the mesh size of the finite element discretization. In general, we are not able to compute $\lambda_1$ of the saddle point system and thus obtain a more precise estimate of $\eta$.
For the smallest three test problems above, we are however able to determine $\lambda_1$ using Matlab and we can compare the previous choice to the optimal value.
From \cref{tab:condnum1,tab:condnum}, it seems that the choice of $\eta = \frac{1}{\gamma}\|\Wm\|_1$ leads to smaller values than needed for \cref{thm:eta}. Using $\eta$ as given in \cref{tab:condnum}, the number of iterations stays at 8 for Problems 2 and 3. 
A short study on the possible choice of $\eta$ in \cref{tab:prob4eta} suggests similar behavior for problem 4. Furthermore, \cref{tab:prob4eta} shows that the number of iterations decreases with increasing $\eta$ and that the modification of the (1,1)-block is a major factor determining the speed of convergence of the GKB method. Note that this behavior agrees with Theorem 2.1.
\\
\\

\begin{table}[htb]
\caption{Parameter $\eta$ and condition numbers of matrices}
\begin{center}
\begin{tabular}{|c|c|c|c|c|c|}
\hline
name & $\lambda_1 $&  $\eta$& $\kappa(\tilde{\Am})^2$ & iter & $\kappa(\Mm)$\\
\hline
Prob. 1 & 0.06 &  17 & 1.978 & 10 & $1.5 \cdot 10^6$ \\
Prob. 2 & 7.5e-3 &  133 &1.995 & 8 & $1.0 \cdot 10^8$\\
Prob. 3 & 2.8e-3& 357 & 1.995& 8 & $1.2 \cdot 10^9$ \\
\hline
\end{tabular}
\end{center}
\label{tab:condnum}
\end{table}

\begin{table}[htb]
\caption{Different choices for $\eta$ for problem 4,  $\epsilon_{GKB}$=1e-5 and $d=5$}
\begin{center}
\begin{tabular}{|c|c|c|c|}
\hline
 $\eta$& \#iter & $\frac{\|\uv - \uv_{dir}\|_{\Mm}}{\|\uv_{dir}\|_{\Mm}}$ & $\frac{\|\pv - \pv_{dir}\|_{2}}{\|\pv_{dir}\|_{2}}$\\
\hline
 0 (\Mm=\Wm) &327 & $8.83 \cdot 10^{-6}$ &$7.09 \cdot 10^{-6}$\\
 1 & 29 & $1.02 \cdot 10^{-7}$  & $8.98 \cdot 10^{-8}$ \\
 17 & 13  & $7.59 \cdot 10^{-11}$& $1.35 \cdot 10^{-10}$ \\
133 & 9  & $3.41 \cdot 10^{-10}$ &  $2.53 \cdot 10^{-10}$\\
357& 8  & $4.57 \cdot 10^{-10}$  & $7.88 \cdot 10^{-10}$ \\
 \hline
\end{tabular}
\end{center}
\label{tab:prob4eta}
\end{table}

\subsection{Example: Prestressed concrete}\label{sec:prest_concrete}

As our second set of examples, we consider a simple model of a concrete block with embedded pretension cables. The block is clamped on its lateral faces and submitted to a constant pressure on its top face. All materials are elastic. \Cref{fig:PSB:simpleModel} presents a projected view to the 2D surface: the orange points are the concrete nodes and the gray points are the cable nodes. The cable nodes are only constrained by linear relationships with the concrete nodes, so that the displacement of the cables included in a given concrete element is a linear combination of the displacement of the concrete nodes, $ \uv_{\mbox{cables}}=\sum_{i=0}^{4} a_i \uv^x_{\mbox{concrete}} + b_i \uv_{\mbox{concrete}}^y$. The vectors $a$ and $b$ are the barycentric coordinates of the cable node with respect to the concrete element \cite{Pe2011}.

\begin{figure}[h]
\begin{center}
\includegraphics[scale=0.2]{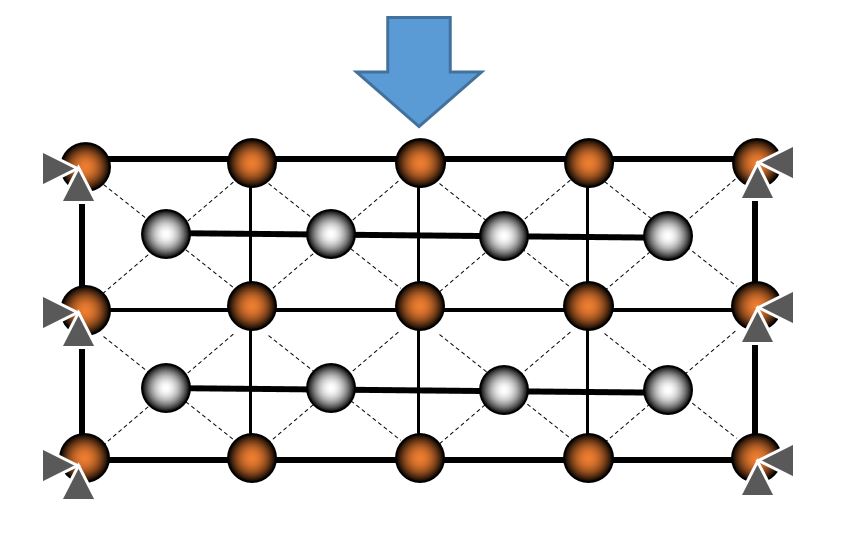}
\end{center}
\caption{Simple model of prestressed concrete}
\label{fig:PSB:simpleModel}
\end{figure}

We first extract the $\Wm$ and $\Am$ submatrices as already described in \cref{ex1:cylinder}. For the purpose of illustrating the particular matrix structure, we apply a permutation to sort the matrix entries in the (1,1)-block with respect to the size of the diagonal elements of $\Wm$, starting from the smallest to the largest. Second, we apply a column permutation to the constraint block $\Am$ (and the respective row permutation for $\Am^T$) to obtain the diagonal part in the upper $n\times n$ block, as shown in \cref{fig:PSB:matrix}. The augmented matrix exhibits particular features. The (1,1)-block contains rows and columns with only zero entries. However, the non-singularity of the full system \cref{eqn:augsys} is ensured, since \cref{eqn:WKerAt} is satisfied. 

\subsubsection{Numerical experiments}
Owing to the singular (1,1)-block, the GKB algorithm as introduced in \cref{sec:GGKB} cannot be directly applied to this problem class. We thus rely on the augmented Lagrangian approach and choose $\eta = \|\Wm\|_{1}$. \cref{eqn:WKerAt} now ensures that the (1,1)-block of the augmented system is non-singular. However, also for these shifted matrices, we do not obtain satisfactory results with the GKB algorithm because of the unfavorable scaling of the matrices when generated by code\_aster. The algorithm converges, the solution however exhibits oscillations. As described in \cref{sec:extr_data}, the constraint matrices are multiplied by the factor $\gamma = \frac{1}{2}(\min\Wm_{ii}+\max\Wm_{ii})$ to obtain a good equilibrium of the augmented system. We undo this multiplication in our numerical experiments and divide the augmented system \cref{eqn:extractedSystem} by $\gamma$. The right-hand sides are provided by code\_aster and the exact solutions are obtained for comparison by solving \cref{eqn:doubleLg} with a direct solver.\\
\\
Numerical results are presented in \cref{tab:precst_concrete}. We use the lower bound estimate as stopping criterion and choose the tolerance as $\tau = 10^{-5}$ and $d = 5$. The algorithm shows excellent convergence properties. Although the result of \cref{thm:eta} is not applicable to this case, the number of iterations until convergence stays constant at 8 and is bounded with increasing problem size. Indeed, the energy error is already smaller than the tolerance after only 3 iterations, but we recall that the lower bound estimate for the iterate $\uv^3$ is only computed at iteration $3+d$. The bound for the smallest singular value of $\Bm$, necessary for the upper bound estimate, has been obtained experimentally as $a = 0.2$. The convergence of the energy error and the lower and upper bound estimates are presented in \cref{fig:PSB:errUL12,fig:PSB:errUL3}.

\begin{table}
\caption{Example prestressed concrete: Golub-Kahan convergence for $\epsilon_{GKB}$=1e-5 and $d=5$}
\begin{center}
\begin{tabular}{|c|c|c|c|c|c|c|}
\hline
name & m & n & \#Iter & $\frac{\| \uv - {\uv}^{(k)} \|_{\Mm}}{\| \uv \|_{\Mm}} $ & $\frac{\| \uv - {\uv}^{(k)} \|_{2}}{\|\uv\|_2}$ &  $\frac{\| \pv -{\pv}^{(k)} \|_{2}}{\|\pv\|_2}$\\
\hline
Prob1 & 498 & 258 & 9 & 9.6e-13 & 9.5e-13  & 2.0e-12\\
Prob2 & 3207 & 1590 & 9 & 3.2e-12 & 3.1e-12 & 9.2e-12\\
Prob3 & 23043 & 11382 & 9 & 5.0e-11 & 5.0e-11 &  4.9e-11\\
\hline
\end{tabular}
\end{center}
\label{tab:precst_concrete}
\end{table}

\begin{figure}
\begin{center}
\includegraphics[width=9.0cm]{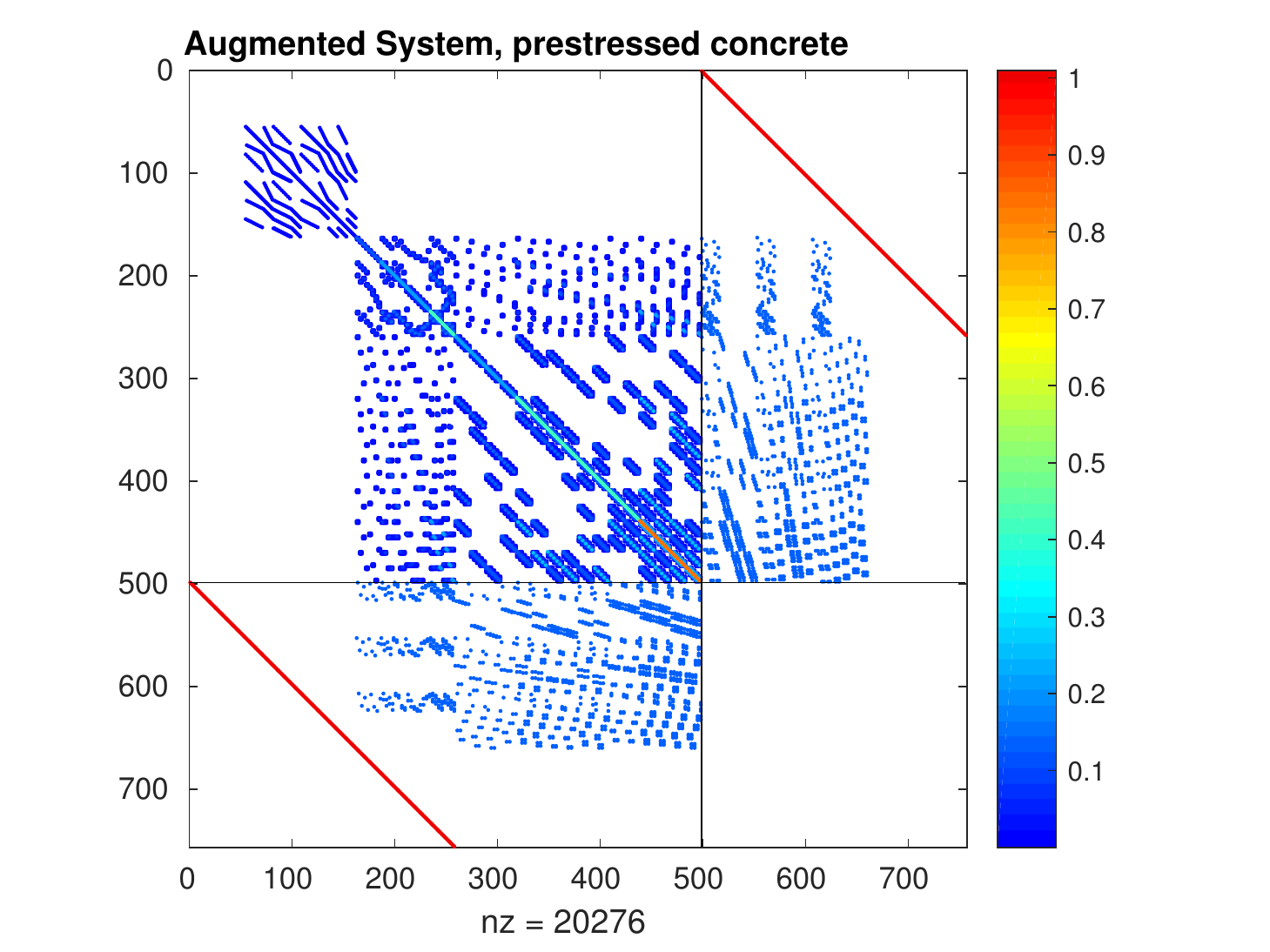}
\end{center}
\caption{Augmented system for prestressed block example.}
\label{fig:PSB:matrix}
\end{figure}

\begin{figure}
\begin{center}
\includegraphics[width=6.4cm]{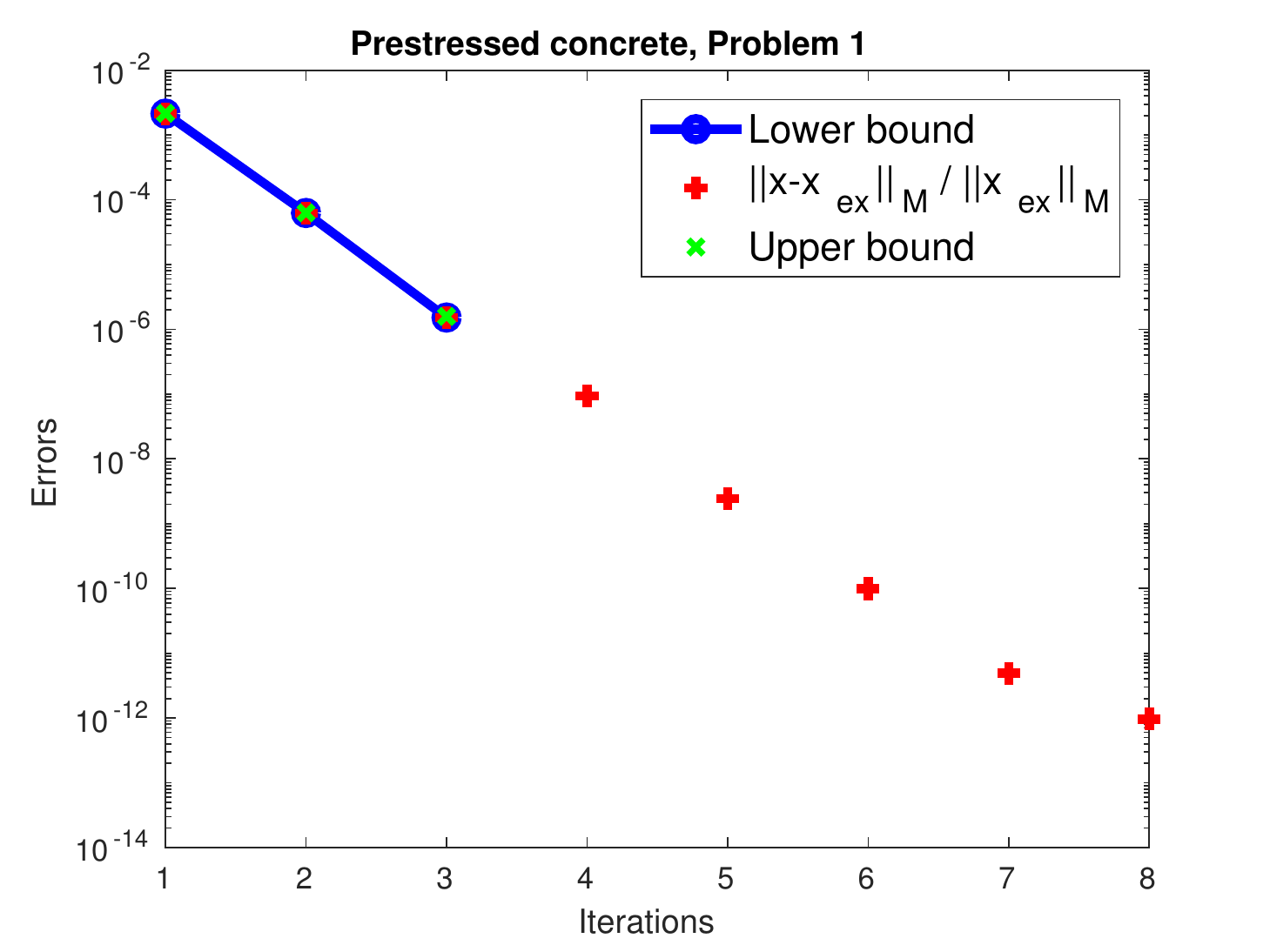}
\includegraphics[width=6.4cm]{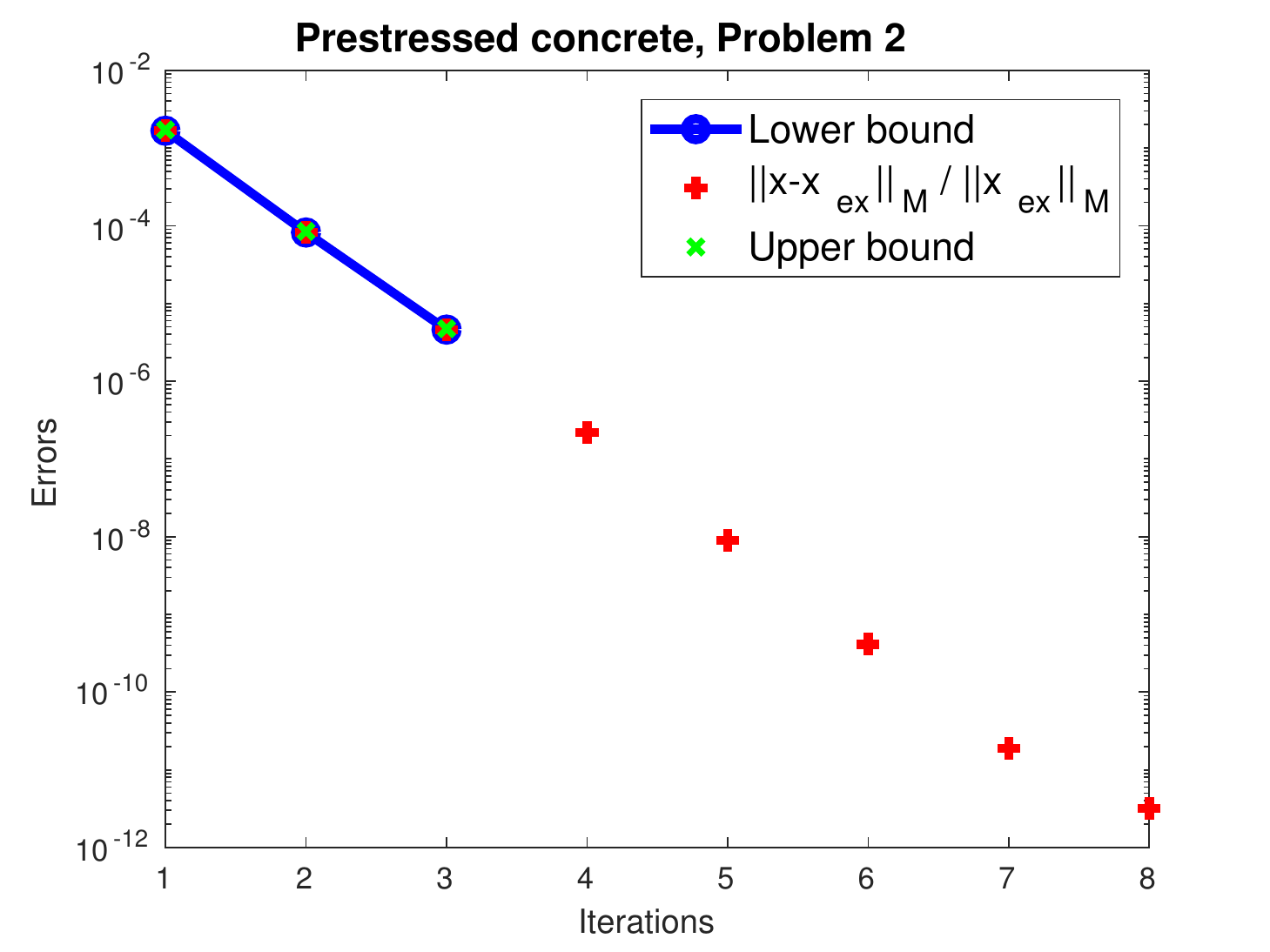}
\end{center}
\caption{GKB convergence for Problem 1 and 2.}
\label{fig:PSB:errUL12}
\end{figure}

\begin{figure}
\begin{center}
\includegraphics[width=6.4cm]{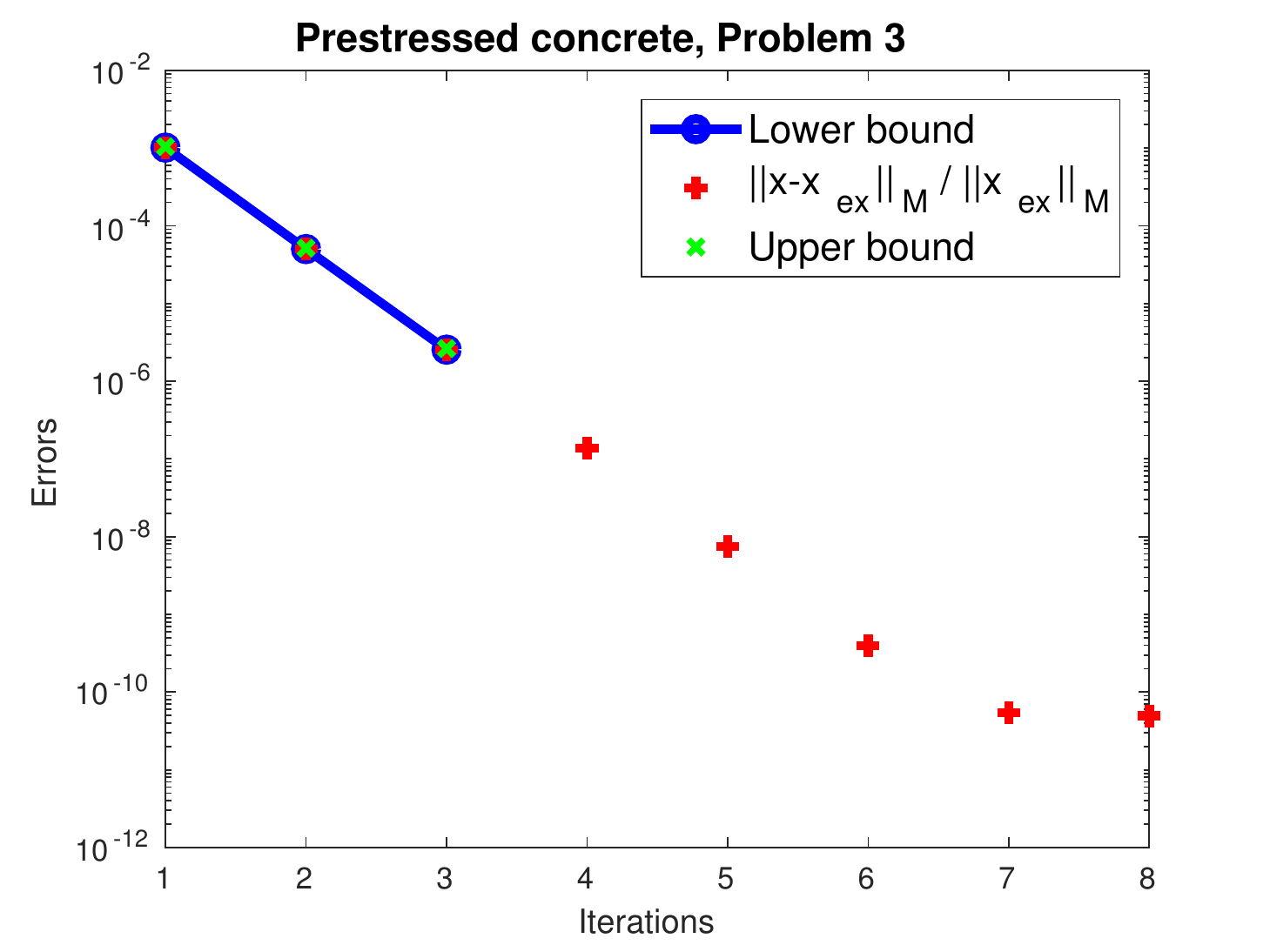}
\end{center} 
\caption{Convergence of generalized GKB method for Problem 3.}
\label{fig:PSB:errUL3}
\end{figure}

\section{Large scale example and parallel implementation}\label{sec:containment}
In this example, we study a critical industrial application, the structural analysis of the reactor containment building of a nuclear power plant. The structure is set under compression during the construction phase, such that it resists better outer influences. The containment building additionally consists of an outer shell layer. The model thus requires the coupling of three dimensional elements (the concrete), two dimensional elements (the outer shell) and one dimensional elements representing the metallic prestressing cables (\cref{fig:containment_building}). The underlying equations for each material are those of linear elasticity.

\begin{figure}
\begin{center}
\includegraphics[width=13.0cm]{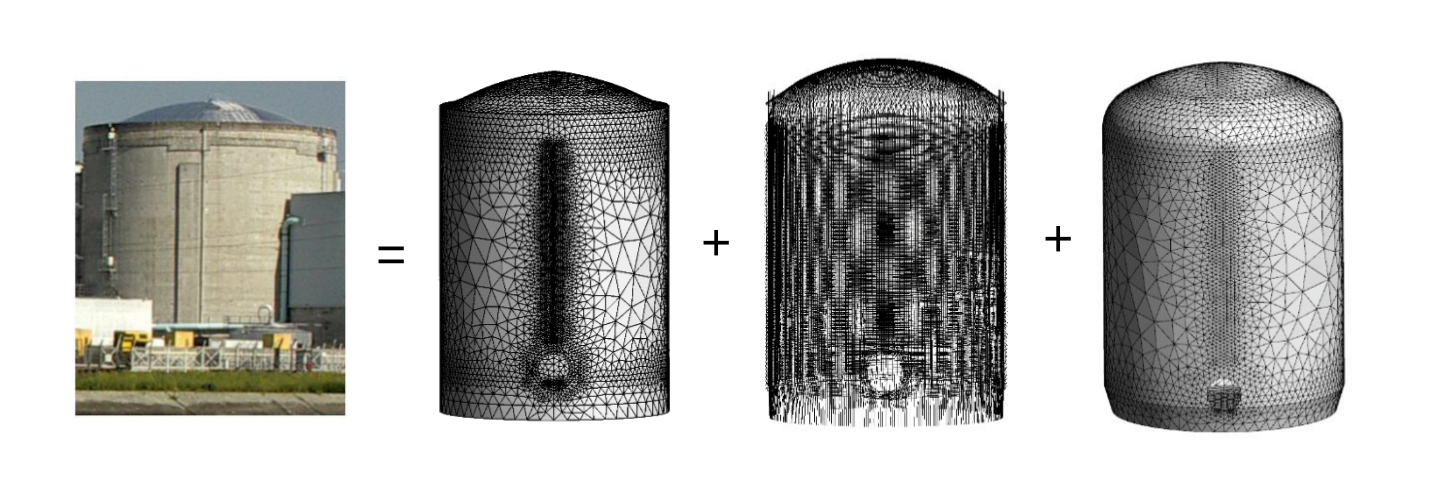}
\end{center}
\caption{Modeling of a containment building.}
\label{fig:containment_building}
\end{figure}

\subsection{Numerical Experiments}
The matrix is generated by code\_aster. The discretization is illustrated in \cref{fig:containment_building} and the blocks are of size $m = 283797$ and $n = 158928$. The number of constraints is thus more than 50\% of the number of physical degrees of freedom. We apply the permutations as explained in the above-mentioned example in \cref{sec:prest_concrete} and obtain the matrix presented in \cref{fig:containment_matrix}. The augmented system contains row and columns with only zero entries in the (1,1)-block, but again \cref{eqn:WKerAt} holds and the nonsingularity of \cref{eqn:augsys} is ensured by the constraint matrix.\\
\\
We implement the Golub-Kahan bidiagonalization method in Julia \cite{Bezanson2014} and we use the interface to the parallel direct solver MUMPS \cite{MUMPS:1} from the JuliaSmoothOptimizers package \footnote{\url{https://github.com/JuliaSmoothOptimizers}} to solve the inner linear system. The factorization of the system matrix is done once. The right-hand side is provided by code\_aster and the exact solution is obtained for comparison by solving \cref{eqn:doubleLg} with MUMPS. As in the previous example, we scale the augmented system \cref{eqn:extractedSystem} with the factor $\gamma = \frac{1}{2}(\min\Wm_{ii} + \max\Wm_{ii})$. The GKB method is not directly applicable to the augmented system \cref{eqn:augsys} with a singular (1,1)-block. For this reason, but also to obtain an improved convergence for the GKB method, we apply the augmented Lagrangian approach with $\eta = \| \Wm\|_1$. Again we use the tolerance $\tau=10^{-5}$ for the lower bound stopping criterion of the GKB method and $d=5$. We apply the algorithm to the unpermuted system as it is obtained from code\_aster. \\

\begin{figure}
\begin{center}
\includegraphics[width=9.0cm]{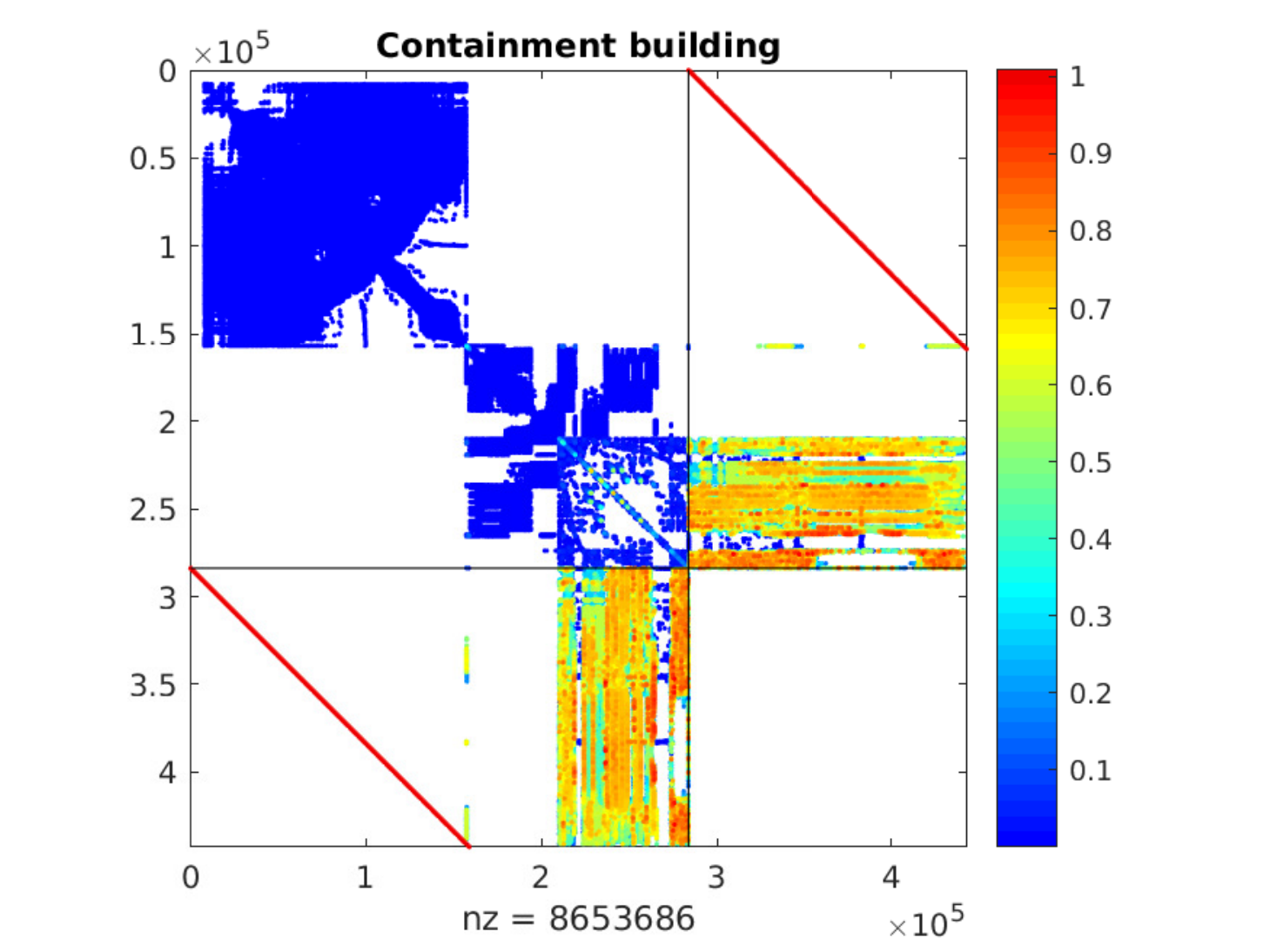}
\end{center}
\caption{Augmented system after permutation and scaling. }
\label{fig:containment_matrix}
\end{figure}

\begin{figure}
\begin{center}
\includegraphics[width=9.0cm]{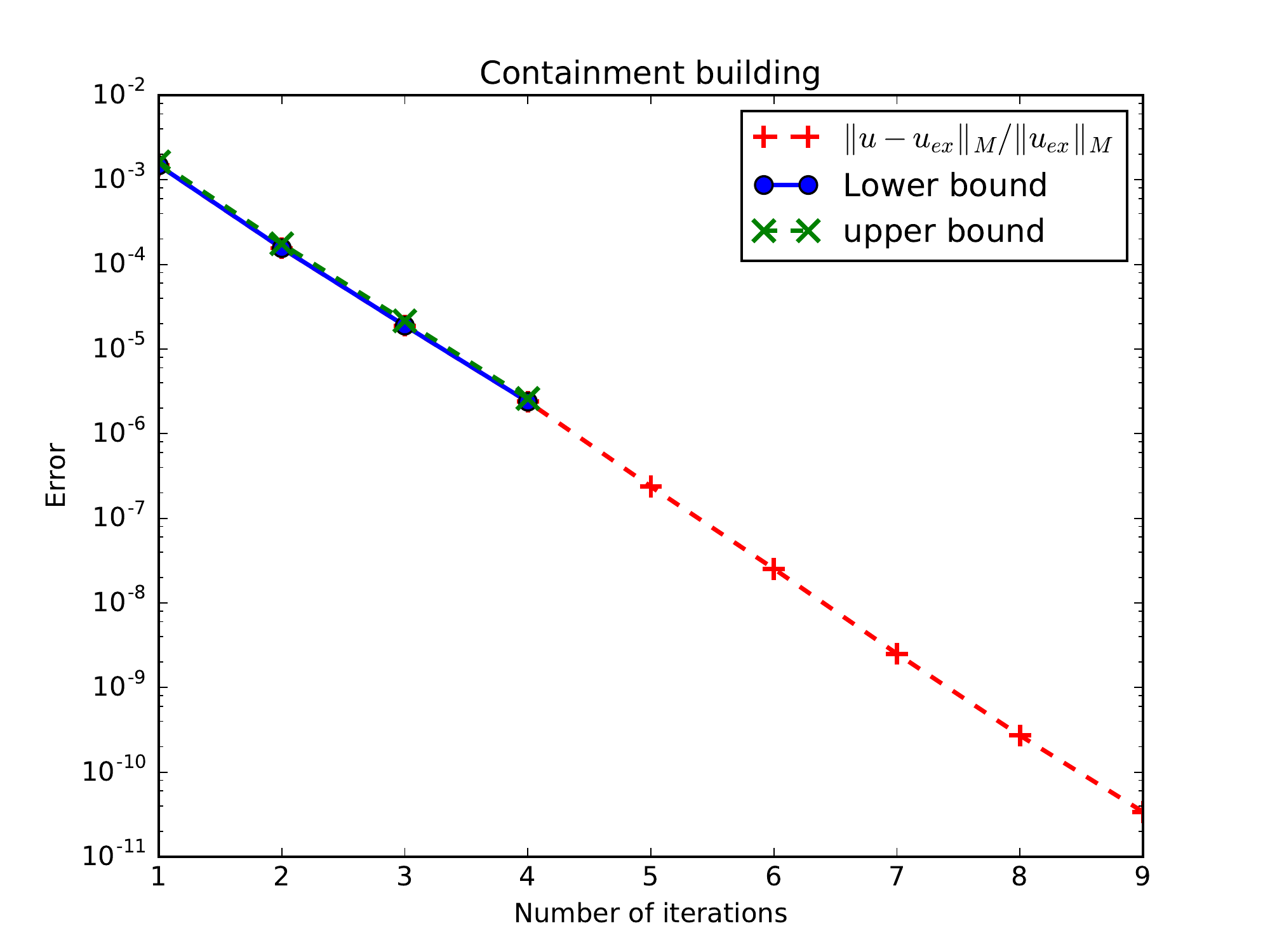}
\end{center}
\caption{GKB convergence.}
\label{fig:containment}
\end{figure}

The algorithm stops after 9 iterations and the relative errors of $\uv$ and $\pv$ are summarized in \cref{tab:containment}. The upper and lower bound estimates are presented in \cref{fig:containment}. Here, the lower bound for the smallest singular values as needed for the upper bound has been estimated numerically as $a = 0.2$.\\
\\
Also for this realistic industrial test case, the GKB iterative method converges after only 9 $(4+d)$ iterations. The final errors obtained in the energy and 2-norm for the solution $\uv$ and also the Lagrange multipliers $\pv$ are remarkably small. In fact, they are by several orders of magnitude better than the required stopping tolerance. Furthermore, we reduce the problem of solving a matrix of size $m + 2n$ (as currently implemented in code\_aster), to solve a linear system of size $m$. We compared the efficiency of the proposed GKB iterative method to solving \cref{eqn:doubleLg} directly with MUMPS under the same conditions. A complete performance analysis of the algorithm is outside the scope of this paper, but a preliminary study shows that in sequential simulations speedups of a factor between 2 and 3 can be observed.
\begin{table}
\caption{Golub-Kahan convergence for the containment building example and $\epsilon_{GKB}$=1e-5 and $d=5$}
\begin{center}
\begin{tabular}{|c|c|c|c|c|c|}
\hline
m & n & \#Iter & $\frac{\| \uv - {\uv}^{(k)} \|_{\Mm}}{\| \uv \|_{\Mm}} $ & $\frac{\| \uv - {\uv}^{(k)} \|_{2}}{\|\uv\|_2}$ &  $\frac{\| \pv -{\pv}^{(k)} \|_{2}}{\|\pv\|_2}$\\
\hline
283797 & 158928 & 9 & 3.39e-11 & 3.8e-11 & 1.4e-9\\
\hline
\end{tabular}
\end{center}
\label{tab:containment}
\end{table}

\section{Conclusions}

In this work, we presented an algorithm based on the Golub-Kahan bidiagonalization method and applied it to problems in structural mechanics. These problems exhibit the difficulty of multi-point constraints imposed on the discretized finite element formulation. We showed that the GKB algorithm converges in only a few iterations for each of the three classes of test problems. In particular, we confirmed our main result of \cref{thm:eta}: The number of GKB iterations is independent of the discretization size for a given problem, whenever we choose the stabilization parameter $\eta$ appropriately. This has also been true for the example of a block of prestressed concrete, although the leading block is singular and does not satisfy the requirements of \cref{thm:eta}. The errors obtained for the solutions $\uv$ and $\pv$ are remarkably small and since the lower bound of the error at iteration $k$ can only be computed at $k+d$, they  undershoot the required tolerance by several orders of magnitude. Summarizing, the proposed algorithm presents a new alternative to the more commonly used standard iterative solvers and, in particular, the ones provided currently in code\_aster. 
\\
\\
The final example of the reactor containment building is a realistic application. However, the dimensions of the matrices are still relatively small. For other applications, the number of degrees of freedoms might be in the order of millions, when also the inner direct solver MUMPS will no longer be satisfactory. It is thus indispensable to solve the inner linear system defined by $\Mm$ with an iterative scheme, which results in an inner-outer iterative method. The study of such algorithms will be the subject of future work.

\bibliographystyle{siamplain}
\bibliography{ind}

\begin{thebibliography}{10}

\bibitem{Abbas2013}
{\sc M.~Abbas}, {\em La m\'{e}thode des \'{e}l\'{e}ments finis
  isoparam\'{e}triques}, tech. report, EDF R\&D, 2013.
\newblock R3.01.00.

\bibitem{Abel1979}
{\sc J.~F. Abel and M.~S. Shephard}, {\em An algorithm for multipoint
  constraints in finite element analysis}, International Journal for Numerical
  Methods in Engineering, 14, pp.~464--467,
  \url{https://doi.org/10.1002/nme.1620140312},
  \url{https://onlinelibrary.wiley.com/doi/abs/10.1002/nme.1620140312},
  \url{https://arxiv.org/abs/https://onlinelibrary.wiley.com/doi/pdf/10.1002/nme.1620140312}.

\bibitem{MUMPS:1}
{\sc P.~Amestoy, I.~Duff, J.~L'Excellent, and J.~Koster}, {\em A fully
  asynchronous multifrontal solver using distributed dynamic scheduling}, SIAM
  Journal on Matrix Analysis and Applications, 23 (2001), pp.~15--41,
  \url{https://doi.org/10.1137/S0895479899358194},
  \url{https://doi.org/10.1137/S0895479899358194},
  \url{https://arxiv.org/abs/https://doi.org/10.1137/S0895479899358194}.

\bibitem{Ar2013}
{\sc M.~Arioli}, {\em Generalized {G}olub--{K}ahan bidiagonalization and
  stopping criteria}, SIAM Journal on Matrix Analysis and Applications, 34
  (2013), pp.~571--592, \url{https://doi.org/10.1137/120866543},
  \url{https://doi.org/10.1137/120866543},
  \url{https://arxiv.org/abs/https://doi.org/10.1137/120866543}.

\bibitem{Benbow1999}
{\sc S.~Benbow}, {\em Solving generalized least-squares problems with lsqr},
  SIAM Journal on Matrix Analysis and Applications, 21 (1999), pp.~166--177,
  \url{https://doi.org/10.1137/S0895479897321830},
  \url{https://doi.org/10.1137/S0895479897321830},
  \url{https://arxiv.org/abs/https://doi.org/10.1137/S0895479897321830}.

\bibitem{BeGoLi2005}
{\sc M.~Benzi, G.~H. Golub, and J.~Liesen}, {\em Numerical solution of saddle
  point problems}, Acta Numerica, 14 (2005), p.~1–137,
  \url{https://doi.org/10.1017/S0962492904000212}.

\bibitem{Bernardi1989}
{\sc C.~Bernardi, Y.~Maday, and A.~T. Patera}, {\em A new nonconforming
  approach to domain decomposition : the mortar element method}, Nonlinear
  partial equations and their applications,  (1994), pp.~13--51.

\bibitem{Bezanson2014}
{\sc J.~Bezanson, A.~Edelman, S.~Karpinski, and V.~Shah}, {\em Julia: A fresh
  approach to numerical computing}, SIAM Review, 59 (2017), pp.~65--98,
  \url{https://doi.org/10.1137/141000671},
  \url{https://doi.org/10.1137/141000671},
  \url{https://arxiv.org/abs/https://doi.org/10.1137/141000671}.

\bibitem{GoKa1965}
{\sc G.~Golub and W.~Kahan}, {\em Calculating the singular values and
  pseudo-inverse of a matrix}, Journal of the Society for Industrial and
  Applied Mathematics Series B Numerical Analysis, 2 (1965), pp.~205--224,
  \url{https://doi.org/10.1137/0702016}, \url{https://doi.org/10.1137/0702016},
  \url{https://arxiv.org/abs/https://doi.org/10.1137/0702016}.

\bibitem{GoGr2003}
{\sc G.~H. Golub and C.~Greif}, {\em On solving block-structured indefinite
  linear systems}, SIAM Journal on Scientific Computing, 24 (2003),
  pp.~2076--2092, \url{https://doi.org/10.1137/S1064827500375096},
  \url{https://doi.org/10.1137/S1064827500375096},
  \url{https://arxiv.org/abs/https://doi.org/10.1137/S1064827500375096}.

\bibitem{GoMeu2010}
{\sc G.~H. Golub and G.~Meurant}, {\em Matrices, Moments and Quadrature with
  Applications}, Princeton University Press, 2010,
  \url{http://www.jstor.org/stable/j.ctt7tbvs}.

\bibitem{jendele2009}
{\sc L.~Jendele and J.~Červenka}, {\em On the solution of multi-point
  constraints – application to fe analysis of reinforced concrete
  structures}, Computers and Structures, 87 (2009), pp.~970 -- 980,
  \url{https://doi.org/https://doi.org/10.1016/j.compstruc.2008.04.018},
  \url{http://www.sciencedirect.com/science/article/pii/S0045794908001351}.
\newblock Computational Structures Technology.

\bibitem{OrAr2017}
{\sc D.~Orban and M.~Arioli}, {\em Iterative Solution of Symmetric
  Quasi-Definite Linear Systems}, SIAM Spotlights, Society for Industrial and
  Applied Mathematics, 2017,
  \url{https://books.google.fr/books?id=Z2xUMQAACAAJ}.

\bibitem{PaSa1982}
{\sc C.~C. Paige and M.~A. Saunders}, {\em {LSQR}: An algorithm for sparse
  linear equations and sparse least squares}, ACM Trans. Math. Softw., 8
  (1982), pp.~43--71, \url{https://doi.org/10.1145/355984.355989},
  \url{http://doi.acm.org/10.1145/355984.355989}.

\bibitem{Pe2011}
{\sc J.~Pellet}, {\em Conditions de liaison de corps solide}, tech. report, EDF
  R\&D, 2011.
\newblock R3.03.02.

\bibitem{Pe2011_01}
{\sc J.~Pellet}, {\em Dualisation des conditions aux limites}, tech. report,
  EDF R\&D, 2011.
\newblock R3.03.01.

\bibitem{Saad2003}
{\sc Y.~Saad}, {\em Iterative Methods for Sparse Linear Systems}, Society for
  Industrial and Applied Mathematics, second~ed., 2003,
  \url{https://doi.org/10.1137/1.9780898718003},
  \url{https://epubs.siam.org/doi/abs/10.1137/1.9780898718003},
  \url{https://arxiv.org/abs/https://epubs.siam.org/doi/pdf/10.1137/1.9780898718003}.

\bibitem{stgeorges98}
{\sc P.~Saint-Georges, Y.~Notay, and G.~Warzée}, {\em Efficient iterative
  solution of constrained finite element analyses}, Computer Methods in Applied
  Mechanics and Engineering, 160 (1998), pp.~101 -- 114,
  \url{https://doi.org/https://doi.org/10.1016/S0045-7825(97)00286-7},
  \url{http://www.sciencedirect.com/science/article/pii/S0045782597002867}.

\bibitem{Sa1995}
{\sc M.~A. Saunders}, {\em Solution of sparse rectangular systems using {LSQR}
  and {CRAIG}}, BIT Numerical Mathematics, 35 (1995), pp.~588--604,
  \url{https://doi.org/10.1007/BF01739829},
  \url{https://doi.org/10.1007/BF01739829}.

\bibitem{Sa1997}
{\sc M.~A. Saunders}, {\em Computing projections with {LSQR}}, BIT Numerical
  Mathematics, 37 (1997), pp.~96--104,
  \url{https://doi.org/10.1007/BF02510175},
  \url{https://doi.org/10.1007/BF02510175}.

\end{thebibliography}

\end{document}